\documentclass[11pt]{article}
\usepackage[square]{natbib}
\AtBeginDocument{}
\usepackage{tikz}

\usetikzlibrary{matrix,decorations.pathreplacing}

\usepackage{xcolor}

\usepackage{easybmat}
\usepackage{multirow,bigdelim}
\usepackage{amsmath}
\usepackage{amsthm}
\usepackage{amssymb}
\usepackage{algorithm}
\usepackage{subcaption}
\usepackage{color}
\usepackage[english]{babel}
\usepackage{graphicx}
\usepackage{wrapfig,epsfig}
\usepackage{epstopdf}
\usepackage{url}
\usepackage{graphicx}
\usepackage{color}
\usepackage{epstopdf}
\usepackage{algpseudocode}
\usepackage{scrextend}
\usepackage[T1]{fontenc}
\usepackage{bbm}
\usepackage{comment}
\usepackage{xspace}

\usepackage{tikz}
\usepackage{hyperref}  
\hypersetup{colorlinks=true,citecolor=blue,linkcolor=blue} 
\usetikzlibrary{arrows}
\usepackage[margin=1in]{geometry}
\author{
    Raghavendra Addanki\thanks{UMass Amherst. \texttt{raddanki@cs.umass.edu}.}
  \and
  Andrew McGregor\thanks{UMass Amherst. \texttt{mcgregor@cs.umass.edu} }
  \and
  Cameron Musco\thanks{UMass Amherst. \texttt{cmusco@cs.umass.edu} }
}

\date{}
\title{Intervention Efficient Algorithms for \\ ~Approximate Learning of Causal Graphs}
\usepackage{times}

\definecolor{mygreen}{RGB}{80,180,0}
\definecolor{b2}{RGB}{51,153,255}
\definecolor{mycy2}{RGB}{255,51,255}

\usepackage{lineno}

\usepackage{etoolbox}

\makeatletter

\newcommand{\define}[4][ignore]{%
  \ifstrequal{#1}{ignore}{}{
  \@namedef{thmtitle@#2}{#1}}%
  \@namedef{thm@#2}{#4}%
  \@namedef{thmtypen@#2}{lemma}%
  \newtheorem{thmtype@#2}[theorem]{#3}%
  \newtheorem*{thmtypealt@#2}{#3~\ref{#2}}%
}

\newcommand{\state}[1]{%
  \@namedef{curthm}{#1}
  \@ifundefined{thmtitle@#1}{
  \begin{thmtype@#1}
    }{
  \begin{thmtype@#1}[\@nameuse{thmtitle@#1}]
  }
    \label{#1}
    \@nameuse{thm@#1}
  \end{thmtype@#1}
  \@ifundefined{thmdone@#1}{
  \@namedef{thmdone@#1}{stated}%
  }{}
}

\newcommand{\restate}[1]{%
  \@namedef{curthm}{#1}
  \@ifundefined{thmtitle@#1}{
    \begin{thmtypealt@#1}
    }{
  \begin{thmtypealt@#1}[\@nameuse{thmtitle@#1}]
  }
    \@nameuse{thm@#1}
  \end{thmtypealt@#1}
  \@ifundefined{thmdone@#1}{
  \@namedef{thmdone@#1}{stated}%
  }{}
}

\newcommand{\thmlabel}[1]{
  \@ifundefined{thmdone@\@nameuse{curthm}}{\label{#1}
    }{\tag*{\eqref{#1}}}
}
\makeatother

\usepackage{tikz}
\usetikzlibrary{matrix,decorations.pathreplacing}
\usepackage{xcolor}

\usepackage{easybmat}
\usepackage{multirow,bigdelim}
\usepackage{amsmath}
\usepackage{amssymb}
\usepackage{algorithm}
\usepackage{algpseudocode}
\usepackage{color}
\usepackage[english]{babel}
\usepackage{graphicx}
\usepackage{epstopdf}
\usepackage{url}
\usepackage{graphicx}
\usepackage{color}
\usepackage{epstopdf}
\usepackage{scrextend}
\usepackage[T1]{fontenc}
\usepackage{bbm}
\usepackage{comment}

\usepackage{tikz}
\usepackage{hyperref}  
\hypersetup{colorlinks=true,citecolor=blue,linkcolor=blue} 
\usetikzlibrary{arrows}
\graphicspath{{./figs/}}
\newcommand{\indep}{\rotatebox[origin=c]{90}{$\models$}}
\newcommand*{\centernot}{%
  \mathpalette\@centernot
}

\def\@centernot#1#2{%
  \mathrel{%
    \rlap{%
      \settowidth\dimen@{$\m@th#1{#2}$}%
      \kern.5\dimen@
      \settowidth\dimen@{$\m@th#1=$}%
      \kern-.5\dimen@
      $\m@th#1\not$%
    }%
    {#2}%
  }%
}

\newcommand{\notindep}{\not\!\perp\!\!\!\perp}
\newcommand{\return}{\textbf{return}}
\date{}

\newcommand\widebar[1]{\mathop{\overline{#1}}}
\newtheorem{theorem}{Theorem}[section]
\newtheorem{lemma}[theorem]{Lemma}
\newtheorem{definition}[theorem]{Definition}

\newtheorem{proposition}[theorem]{Proposition}
\newtheorem{corollary}[theorem]{Corollary}

\newtheorem{claim}[theorem]{Claim}

\newcommand{\lat}{L}
\newcommand{\Edg}{\mathcal{E}}
\newcommand{\NearMIS}{\textsc{Near-MIS}}
\newcommand{\IndepSet}{\textsc{Independent-Set}}

\newcommand{\wh}{\widehat}
\newcommand{\wt}{\widetilde}

\newcommand{\eps}{\epsilon}

\newcommand{\R}{\mathbb{R}}

\newcommand{\norm}[1]{\left\lVert#1\right\rVert}
\renewcommand{\varepsilon}{\epsilon}
\renewcommand{\tilde}{\wt}
\renewcommand{\hat}{\wh}

\DeclareMathOperator{\OPT}{OPT}
\DeclareMathOperator{\ALG}{ALG}

\definecolor{mygreen}{RGB}{80,180,0}
\definecolor{b2}{RGB}{51,153,255}
\definecolor{mycy2}{RGB}{255,51,255}

\def \doo {\mathrm{do}}

\def \e {\mathrm{e}}
\def \Anc {\mathrm{Anc}}
\def \GGG {\mathcal{G}}

\makeatletter
\newcommand*{\RN}[1]{\expandafter\@slowromancap\romannumeral #1@}
\makeatother
\usepackage{lineno}

\begin{document}

\maketitle

\begin{abstract}
We study the problem of learning the causal relationships between a set of observed variables in the presence of latents, while minimizing the cost of interventions on the observed variables. We assume access to an undirected graph $G$ on the observed variables whose edges represent either all direct causal relationships or, less restrictively, a superset of causal relationships (identified, e.g., via conditional independence tests or a domain expert). 
Our goal is to recover the directions of all causal or ancestral relations in $G$, via a minimum cost set of interventions.

It is known that constructing an exact minimum cost intervention set  for an arbitrary graph $G$ is NP-hard. We further argue that, conditioned on the hardness of approximate graph coloring, no polynomial time algorithm can achieve an approximation factor better than $\Theta(\log n)$, where $n$ is the number of observed variables in $G$. To overcome this limitation, we introduce a bi-criteria approximation goal that lets us recover the directions of all but $\epsilon n^2$ edges in $G$, for some specified error parameter $\epsilon > 0$. Under this relaxed goal, we give polynomial time algorithms that achieve intervention cost within a small constant factor of the optimal. Our algorithms combine work on efficient intervention design and the design of low-cost \emph{separating set systems}, with ideas from the literature on graph property testing.

\end{abstract}

\section{Introduction}\label{sec:intro}
Discovering causal relationships is one of the fundamental problems of causality~\citep{pearl}. In this paper, we study the problem of {\em learning a causal graph} where we seek to identify all the causal relations between variables in our system (nodes of the graph).
It has been shown that, under certain assumptions, observational data alone lets us recover {the existence of a causal relationship} between, but not the {direction} of all relationships. To recover the direction, we use the notion of an intervention (or an experiment) described in Pearl's Structural Causal Models (SCM) framework \citep{pearl}.

An intervention requires us to fix a subset of variables to each value in their domain, inducing a new distribution on the free variables. For example, we may intervene to require that some patients in a study follow a certain diet and others do not. As performing interventions is costly, a widely studied goal is to find a minimum set of interventions for learning the causal graph \citep{shanmugam2015learning}. This goal however does not address the fact that interventions may have different costs. For example, interventions that fix a higher number of variables will be more costly. Additionally, there may be different intervention costs associated with  different variables. For example, in a medical study, intervening on certain variables might be impractical or unethical.~\cite{hyttinen2013experiment} address the need for such cost models and give results for the special case of learning the directions of \emph{complete graphs} when the cost of an intervention is equal to the number of variables contained in the intervention. Generalizing this notion, we study a \emph{linear cost model} where the cost of an intervention on a set of variables is the sum of (possibly non-uniform) costs for each variable in the set. This model was first introduced in~\cite{icml17} and has received recent attention~\citep{neurips18, AKMM2020}.

Significant prior work on efficient intervention design assumes \emph{causal sufficiency}, i.e., there are no unobserved (latent) variables in the system. In this setting,  there is an exact characterization of the interventions required to learn the causal graph, using the notion of \emph{separating set systems}~\citep{shanmugam2015learning, eberhardt2007causation}.

Recently, the problem of learning the causal graph with latents using a minimum number of interventions has received considerable attention with many known algorithms that depend on various properties of the underlying causal graph~\citep{neurips17, AKMM2020, kocaoglu2019characterization}. However, the intervention sets used by these algorithms contain a large number of variables, often as large as $\Omega(n)$, where $n$ is the number of observable variables. Thus, they are generally not efficient in the linear cost model. 
Some work has considered efficient intervention design in the linear cost model for recovering the ancestral graph containing all indirect causal relations \citep{AKMM2020}.  Other algorithms such as IC$^*$ and FCI with running times exponential in the size of the graph, aim to learn the causal graph in the presence of latents using only observational data; however, they can only learn a part of the entire causal graph~\citep{ic,spirtes2000causation}.

\subsection{Our Results}
In order to address the shortcomings when there are latents, we consider two settings. In the first setting, we assume that we are given an undirected graph  that contains all causal relations between observable variables, but must identify their directions. This undirected graph may be obtained, e.g., by running algorithms that identify conditional dependencies and consulting a domain expert to identify causal links. In the second setting, we study a relaxation where we are given a supergraph $H$ of $G$ containing all causal edges and other additional edges which need not be causal. The second setting is less restrictive, modeling the case where we can ask a domain expert or use observational data to identify a \textit{superset} of possible causal relations. 

From $H$ we seek to recover edges of the ancestral graph\footnote{We note that \emph{ancestral graph} defined here and in~\cite{neurips17, AKMM2020} is slightly different from the widely used notion from~\cite{richardson2002}.} of $G$, a directed graph containing all causal \emph{path} relations between the observable variables. Depending on the method by which $H$ is obtained, it may have special properties that can be leveraged for efficient intervention design. For example, if we use FCI/IC$^*$~\citep{spirtes2000causation} to recover a partial ancestral graph from observational data, the remaining undirected edges form a chordal graph~\citep{zhang2008causal}. Past work has also considered the worst case when $H$ is the complete graph \citep{AKMM2020}. In this work, we do not assume anything about how $H$ is obtained and thus give results holding for general graphs.

In both settings, we show a connection to separating set systems. Specifically, to solve the recovery problems it is necessary and sufficient to use a set of interventions corresponding to a separating set system when we are given the undirected causal graph $G$ and a \textit{strongly} separating set system when we are given the supergraph $H$. A separating set system is one in which each pair of nodes connected by an edge is separated by at least one intervention -- one variable is intervened on and the other is free. A strongly separating set system requires that every connected edge $(u, v)$ is separated by two interventions -- there exists a intervention including $u$ but not $v$ and an intervention that includes $v$ but not $u$.

Unfortunately, finding a minimum cost (strongly) separating set system for an arbitrary graph $G$ is NP-hard \citep{neurips18,hyttinen2013experiment}.  We give simple algorithms that achieve $O(\log n)$ approximation and further argue that, 
conditioned on the hardness of approximate graph coloring, no polynomial time algorithm can achieve $o(\log n)$ approximation, where $n$ is the number of observed variables.  

To overcome this limitation, we introduce a \emph{bi-criteria approximation} goal that lets us recover all but $\epsilon n^2$ edges in the causal or ancestral graph, where $\epsilon > 0$ is a specified error parameter. For this goal, it suffices to use a relaxed notion of a set  system, which we show can be found efficiently using ideas from the graph property testing literature \citep{GGR98}. 

In the setting where we are given the causal edges in $G$ and must recover their directions, we give a polynomial time algorithm that finds a set of interventions from which we can recover all but $\epsilon n^2$ edges  with cost at most $\sim 2$ times the optimal cost for learning the full graph. 
Similarly, in the setting of ancestral graph recovery, we show how to recover all but $\epsilon n^2$ edges with intervention cost at most $\sim 4$ times the optimal cost for recovering all edges.

Our result significantly extends the applicability of a previous result~\citep{AKMM2020} that gave a $2$-approximation to the minimum cost strongly separating set system assuming the worst case when the supergraph $H$ is a complete graph. That algorithm does not translate to an approximation guarantee better than $\Omega(\log n)$ for general graphs.

Finally, for the special case when $G$ is a hyperfinite graph~\citep{hassidim2009local} with maximum degree $\Delta$, we give algorithms that obtain approximation guarantees as above, and recover all but $\eps n \Delta$ edges of $G$.

\subsection{Other Related Work}

 There is significant precedent for assumptions on background knowledge in the literature. For example,~\citep{hyttinen2013experiment} and references therein, study intervention design in the same model: a skeleton of possible edges in the causal graph is given via background knowledge, which may come e.g., from domain experts or previous experimental results. Assuming causal sufficiency (no latents), most work focuses on recovering causal relationships based on just observational data. Examples include algorithms like {\em IC}~\citep{pearl} and {\em PC}~\citep{spirtes2000causation}, which have been widely studied~\citep{hauser2014two, hoyer2009nonlinear, heinze2018causal, loh2014high, shimizu2006linear}. It is well-known that to disambiguate a causal graph from an equivalence class of possible causal structures, interventional, rather than just observational data is required~\citep{hauser2012characterization, eberhardt2007interventions, eberhardt2007causation}. 
There is a growing body of recent work devoted to minimizing the number of interventions ~\citep{shanmugam2015learning, neurips17, kocaoglu2019characterization} and costs of intervention~\citep{neurips18, icml17}. 

Since causal sufficiency is often too strong an assumption~\citep{bareinboim2016causal}, many algorithms avoiding the causal sufficiency assumption, such as IC$^*$~\citep{ic} and FCI~\citep{spirtes2000causation}, and using just observational data have been developed. 
 There is a growing interest in optimal intervention design in this setting~\citep{silva2006learning, hyttinen2013experiment, parviainen2011ancestor, neurips17, kocaoglu2019characterization}. 

\section{Preliminaries}\label{sec:prelim}

\smallskip
\noindent \textbf{Causal Graph Model}.
Following the SCM framework \citep{pearl}, we represent a set of random variables by $V \cup \lat$ where $V$ contains the endogenous (observed) variables that can be measured and $\lat$ contains the exogenous (latent) variables that cannot be measured. 
We define a directed causal graph $\GGG = \GGG(V \cup \lat, \Edg)$ on these variables where an edge corresponds to a causal relation between the corresponding variables: a directed edge  $(v_i, v_j)$ indicates that  $v_i$ causes $v_j$.

We assume that all causal relations belong to one of two categories : (i) $E \subseteq V \times V$ containing direct causal relations between the observed variables and (ii) $E_L \subseteq L \times V$ containing relations from latents to observable variables. Thus, the full edge set of our causal graph is $\Edg = E \cup E_L$. We also assume that every latent $l \in \lat$ influences exactly two observed variables, i.e., $(l,u), (l,v) \in E_L$ and no other edges are incident on $l$. This \emph{semi-Markovian } assumption is widely used in prior work~\citep{neurips17,shpitser2006identification} (see Appendix~\ref{app:disc} for a more detailed discussion).  Let $G(V,E)$ denote the subgraph of $\GGG$ restricted to observable variables, referred to as the observable graph. 

Similar to~\citet{neurips17, AKMM2020}, we define \emph{ancestral graph} of $G$ over observable variables $V$, denoted by $\Anc(G)$ as follows : $(v_i, v_j) \in \Anc(G)$ iff there is a directed path from $v_i$ to $v_j$ in $G$ (equivalently in $\mathcal{G}$ due to the semi-Markovian assumption). Throughout we denote $n = |V|$.

\smallskip
\noindent \textbf{Intervention Sets}. Our primary goal is to recover either $G$ or $\Anc(G)$ via interventions on the observable variables. 
We assume the ability to perform an \emph{intervention} on a set of variables $S \subseteq V$ which fixes $S = s$ for each $s$ in the domain of $S$. We then perform a conditional independence test answering  for all $v_i,v_j$ 
\emph{``Is $v_i$ independent of $v_j$ in the interventional distribution $\doo(S = s)$?''} and denote it using $v_i \indep v_j \mid \doo(S)$. Here $\doo(S = s)$ uses Pearl's do-notation to denote the interventional distribution when the variables in $S$ are fixed to $s$.

An \textit{intervention set} is a collection of subsets $\mathcal{S} = \{S_1,\ldots, S_m\}$ that we intervene on in order to recover edges of the observable or ancestral graph. It will also be useful to associate a matrix $L \in \{0,1\}^{n \times m}$ with the collection where the $i$th column is the characteristic vector of set $S_i$, i.e., row entry corresponding to node in $S_i$ is $1$ iff it is present in $S_i$. We can also think of $L$ as a collection of $n = |V|$ length-$m$ binary vectors that indicate which of the $m$ intervention sets $S_1,\ldots,S_m$ each variable $v_i$ belongs to.

As is standard, we assume that $\GGG$ satisfies the {\em causal Markov condition} and assume {\em faithfulness}~\citep{spirtes2000causation}, both in the observational and interventional distributions following~\citep{hauser2014two}. This ensures that conditional independence tests lead to the discovery of true causal relations rather than spurious associations.
 \smallskip
 
\noindent \textbf{Cost Model and Approximate Learning}. In our cost model, each node $u \in V$ has a cost $C(u) \in [1, W]$ for some $W \geq 1$ and the cost of intervention on a set $S \subseteq V$ has the linear form $C(S) = \sum_{u \in S} C(u)$. That is, interventions that involve a larger number of, or more costly nodes, are more expensive. 
Our goal is to find an intervention set $\mathcal{S}$ minimizing $C(\mathcal{S}) = \sum_{S \in \mathcal{S}} \sum_{u \in S} C(u)$, subject to a constraint $m$ on the number of interventions used. This \emph{min cost intervention design} problem was first introduced in~\cite{icml17}.

Letting $L \in \{0,1\}^{n \times m}$ be the matrix associated with an intervention set $\mathcal{S}$, the cost $C(\mathcal{S})$ can be written as $C(L) = \sum_{j=1}^n C(v_j)\cdot \| L(j) \|_1$, where $\norm{L(j)}_1$ is the \emph{weight} of $L$'s $j^{th}$ row, i.e., the number of $1$'s in that row or the number of interventions in which $v_j$ is involved.

We study two variants of causal graph recovery, in which we seek to recover the observable graph $G$ or the ancestral graph $\Anc(G)$.
We say that an intervention set $\mathcal S$ is $\alpha$-{optimal} for a given recovery task if $C(\mathcal S) \le \alpha \cdot C(\mathcal{S}^*)$, where $\mathcal{S}^*$ is the minimum cost intervention set needed for that task. For both recovery tasks we consider a natural approximate learning guarantee: 

\begin{definition}[$\epsilon$-Approximate Learning]\label{def:approx} An algorithm $\epsilon$-approximately learns $G(V, E)$ (analogously, $\Anc(G)$) if it identifies the directions of a subset $\widetilde{E} \subseteq E$ of edges 
with $|E \setminus \widetilde{E} | \leq \epsilon n^2$.\label{def:eps_learn}
\end{definition}
Generally, we will seek an intervention set $\mathcal{S}$ that lets us  $\epsilon$-approximately learn $G$ or $\Anc(G)$, and which has cost bounded in terms of $\mathcal{S}^*$, the minimum cost intervention set needed to \emph{fully} learn the graph. In this sense, our algorithms are bicriteria approximations.

\smallskip
\noindent \textbf{Independent Sets}. Our intervention set algorithms will be based on finding large independent sets of variables, that can be included in the same intervention sets, following the approach of \cite{neurips18}. Given $G(V, E)$, a subset of vertices $Z \subseteq V$ forms an independent set if there are no edges between any vertices in $Z$, i.e., $E[Z] = \emptyset$ where $E[Z]$ is set of edges in the sub-graph induced by $Z$.
Given a cost function $C : V \rightarrow \mathbb{R}^+$, an independent set $Z$ is a maximum cost independent set (MIS) if $C(Z) = \sum_{u \in Z} C(u)$ is maximized over all independent sets in $G$. Since finding MIS is hard~\citep{cormen2009introduction},  we will use the following two notions of a MIS, with the first often referred to as simply \NearMIS, in our approximate learning algorithms :

\begin{definition}[$(\gamma, \epsilon)$-\NearMIS]\label{def:nearMIS}
A set of nodes $S \subseteq V$ is a $(\gamma,\epsilon)$-near-MIS in $G = (V,E)$ if $C(S) \geq (1-\gamma) C(T)$ and $|E[S]| \leq \epsilon n^2$ where $T$ is a maximum cost independent set (MIS) in $G$.
\end{definition}

\begin{definition}[$(\rho, \gamma, \epsilon)$-Independent-Set]\label{def:rhoepsis}
A set of nodes $S \subseteq V$ is a $(\rho, \gamma, \epsilon)$-independent-set in $G = (V,E)$ if $C(S) \geq \rho(1-\gamma) \cdot C(V)$ and $|E[S]| \leq \epsilon n^2$.
\end{definition}

\section{Separating Set Systems}\label{sec:ss}

Our work focuses on two important classes of intervention sets which we show in Sections \ref{sec:graph} and \ref{sec:ancG} are necessary and sufficient for recovering $G$ and $\Anc(G)$ in our setting. Missing details from this section are collected in Appendix~\ref{app:ss}.

\begin{definition}[Separating Set System]\label{def:sss}
For any undirected graph $G(V,E)$, a collection of subsets $\mathcal{S} = \{S_1, \cdots, S_m\}$ of $V$ is a separating set system if every edge $(u,v) \in E$ is separated, i.e., there exists a subset $S_i$ with $u \in S_i$ and $v \notin S_i$ or with $v \in S_i$ and $u \notin S_i$.
\end{definition}
\begin{definition}[Strongly Separating Set System]\label{def:ssss}
For any undirected graph $G(V,E)$, a collection of subsets $\mathcal{S} = \{S_1, \cdots, S_m\}$ of $V$ is a {strongly separating} set system if every edge $(u, v) \in E$ is strongly separated, i.e., there exist two subsets $S_i$ and $S_j$ such that $u\in S_i\setminus S_j$ and $v\in S_j\setminus S_i$.
\end{definition}

For a separating set system, each pair of nodes connected in $G$ must simply have different assigned row vectors in the matrix $L \in \{0,1\}^{n\times m}$ corresponding to $\mathcal S$ (i.e., the rows of $L$ form a valid coloring of $G$). For a strongly separating set system,  the rows must not only be distinct, but one cannot have support which is a subset of the other's. We say that such rows are \emph{non-dominating}: there are distinct $i, j \in [m]$ such that $L(u,i) = L(v,j) = 0$ and $L(u,j) = L(v,i) = 1$. We observe that every \textit{strongly separating set system} must satisfy the non-dominating property (as also observed in Lemma A.9 from~\cite{AKMM2020}).
When $\mathcal{S}$ is a (strongly) separating set system for $G$ we call its associated matrix $L$ a \emph{(strongly) separating matrix} for $G$. 

Finding an exact minimum cost (strongly) separating set system is NP-Hard \citep{neurips18,hyttinen2013experiment} and thus we focus on approximation algorithms. We say the $\mathcal S$ is an $\alpha$-optimal (strongly) separating set system if $C(\mathcal S) \le \alpha \cdot C(\mathcal{S}^*)$, where $\mathcal{S}^*$ is the minimum cost (strongly) separating set system. Equivalently, for matrices $C(L) \le \alpha \cdot C(L^*)$ where $L,L^*$ correspond to $\mathcal{S},\mathcal{S}^*$ respectively.

 Unfortunately, even when approximation is allowed, finding a low-cost set system for an arbitrary graph $G$ is still hard. In particular, we prove a conditional lower bound based on the hardness of approximation for $3$-coloring. Achieving a  coloring for $3$-colorable graphs that uses sub-polynomial colors in polynomial time is a longstanding open problem \citep{wigderson1983improving,blum1997n314, karger1994approximate}, with the current best known algorithm \citep{arora2006new} achieving an approximation factor $O(n^{0.2111})$. Thus Theorem \ref{thm:ss_hardness} shows the hardness of finding near optimal separating set systems, barring a major breakthrough on this classical problem.

\begin{theorem}
Assuming 3-colorable graphs cannot be colored with sub-polynomial colors in polynomial time, there is no polynomial time algorithm for finding an $o(\log n)$-optimal (strongly) separating set system for an arbitrary graph $G$ with $n$ nodes when $m = \beta \log n$ for some constant $\beta > 2$.
\label{thm:ss_hardness}
\end{theorem}
\begin{proof}
We give a proof by contradiction for the case of separating set system A similar proof can be extended to strongly separating set systems. Suppose $G$ is a 3-colorable graph containing $n$ nodes with \textit{unit} costs for every node.  We argue that if there is an $o(\log n)$-optimal algorithm for separating set system then, we can use it to obtain an algorithm for $3$-coloring of $G$ using $n^{o(1)}$ colors, thereby giving a contradiction.

First, we observe that the cost of an optimal separating system on $G$ when $m = \beta \log n$ is at most $n$, as each color class forms an independent set in $G$ and every node in the color class can be assigned a vector of weight at most $1$. Let $\mathcal{A}(G)$ denote the separating set system output by an $\alpha$-optimal algorithm where $\alpha = o(\log n)$. We outline an algorithm that takes as input $\mathcal{A}(G)$ and returns a $n^{o(1)}$-coloring of $G$.
 
We have $C(\mathcal{A}(G)) \leq \alpha C(\mathcal{S}^*)$ where $\mathcal{S}^*$ is an optimal separating set system for $G$. Letting $L$ be the separating matrix associated with $\mathcal{A}(G)$, we thus have
\[
 C(\mathcal{A}(G)) = \sum_{j=1}^n \norm{L(j)}_1 \leq \alpha C(\mathcal{S}^*) \leq \alpha n.
\]
Using an averaging argument, we have that in $\mathcal{A}(G)$, there are at most $\frac{n}{4}$ nodes (denoted by $V \setminus D^{(1)}$) with weight $\norm{L(j)}_1$ more than $4 \alpha$. Consider the remaining $\frac{3n}{4}$ nodes given by $D^{(1)}$. Let $D^{(1)}_j$ denote the nodes that have been assigned weight $j$ by $\mathcal{A}(G)$. For each of the at most ${m \choose j}$ vectors with weight $j$ that are feasible, we create a new color and color each node in $D^{(1)}_j$ using these new colors based on the weight $j$ vectors assigned to the node in $\mathcal{A}(G)$. We repeat this procedure for every weight $j$ in $D^{(1)}$. As the maximum weight of a node in $D^{(1)}$ is $4\alpha$, the total number of colors that we use to color all the nodes of $D^{(1)}$ is 
 \begin{align*}
   \sum_{j=0}^{4\alpha} {m \choose j} \leq \sum_{j=0}^{4\alpha} \frac{m^j}{j!} = \sum_{j=0}^{4\alpha} \frac{(4\alpha)^j}{j!} \left( \frac{m}{4\alpha} \right)^j  
   \leq \e^{4 \alpha} \left( \frac{ m}{4\alpha} \right)^{4\alpha}  
   &\leq 2^{4\alpha \log e   + 4\alpha \log \frac{m}{4\alpha}}\\
 &< \ 2^{4\alpha \log e + \sqrt{4 m \alpha} }\\
 &<\ 2^{o(\log n)+ \sqrt{\log n \cdot o(\log n)} }\\
 &< \ n^{o(1)},
\end{align*}
where the first strict inequality used the fact that $\log \frac{m}{4\alpha} \leq \sqrt{\frac{m}{4\alpha}} \hspace{1ex} \text{ for} \hspace{1ex} \frac{m}{4\alpha} > \ \frac{\beta \log n}{o(\log n)} > \ 32$.
 
 After coloring the nodes of $D^{(1)}$, we remove these nodes from $G$ and run $\alpha$-optimal algorithm $\mathcal{A}$ on the remaining nodes $V \setminus D^{(1)}$. Observing that a sub-graph of a 3-colorable graph is also 3-colorable, we have that the set of nodes obtained by running $\mathcal{A}$ on $V\setminus D^{(1)}$ that have weight at most $4\alpha$ (denoted by $D^{(2)}$) also require at most $n^{o(1)}$ colors. As $|D^{(i)}| \geq \frac{3|V\setminus D^{(i-1)}|}{4}$ for all $i \in \{1, 2,\cdots, \log n\}$, in at most $\log n$ recursive calls to $\mathcal{A}$, we will fully color $G$ using at most $n^{o(1)}  \log n = n^{o(1)}$ colors. Hence, we have obtained a $n^{o(1)}$-coloring of $G$ using an $\alpha$-optimal algorithm when $\alpha = o(\log n)$.

\end{proof}

{ \paragraph{Remark.} The results of Theorem~\ref{thm:ss_hardness} can be extended to any $m$. When $m = o(\log n)$, in our hardness example that uses $3$ colors, any valid separating set system using $m$ interventions would lead to a coloring of the graph using at most $2^m = n^{o(1)}$ colors, i.e., a sub-polynomial number of colors. Thus, even finding a valid separating matrix in this scenario is hard, under our assumed hardness of 3-coloring.}

\medskip

We shall now proceed to discuss a $O(\log n)$ approximation algorithm for finding (strongly) separating set systems. It is easy to check that for a strongly separating set system, every node must appear in at least one intervention (because of non-dominating property), and so the set system has cost as least $\sum_{v \in V} C(v)$. At the same time, with $m \ge 2 \log n$, 
 we can always find a strongly separating set system where each node appears in $\log n$ interventions. In particular, we assign each node to a unique vector with weight $\log n$. Such an assignment is non-dominating and since $\binom{2 \log n}{\log n} \ge n$, is feasible. It achieves cost $C(\mathcal{S}) = \log n \cdot \sum_{v \in V} C(v)$, giving a simple $\log n$-approximation for the minimum cost strongly separating set system problem. For a separating set system, a simple $O(\log n)$-approximation is also achievable by first computing an approximate minimum weight vertex cover and assigning all nodes in its complementary independent set the weight $0$ vector i.e., assigning them to no interventions.

\vspace*{2ex}
\noindent \textbf{A $2\log n$-Approximation Algorithm}. Find a $2$-approximate weighted vertex cover $X$ in $G$ using the classic algorithm from~\cite{williamson2011design}. In $L$, assign zero vector to all nodes of $V\setminus X$; assign every node in $X$ with a unique vector of weight $\log n$ and return $L$.

\vspace*{2ex}

We give a sketch of the arguments involved in proving the approximation ratio of the above algorithm and defer the full details to Appendix~\ref{app:ss}. We observe that all the nodes that are part of maximum cost independent set (complement of minimum weighted vertex cover) are assigned a weight $0$ vector by optimal separating system for $G$. Therefore, the cost of optimal separating set system is at least the cost of minimum cost vertex cover in $G$. As every node is assigned a vector of weight $\log n$ and the cost of vertex cover is at most twice the cost of the minimum weighted vertex cover, we have $C(L) \leq 2\log n \cdot  C(L^*)$. 

\vspace*{2ex}
By Theorem \ref{thm:ss_hardness}, it is hard to improve on the above $O(\log n)$ approximation factor (up to constants). Therefore, we focus on finding relaxed separating set systems in which some variables are not separated. We will see that these set systems still suffice for approximately learning $G$ and $\Anc(G)$ under the notion of Definition \ref{def:approx}.

\begin{definition}[$\epsilon$-(Strongly) Separating Set System]\label{def:espss} For any undirected graph $G(V,E)$, a collection of subsets $\mathcal{S} = \{S_1, \cdots, S_m\}$ of $V$ is an $\epsilon$-separating set system if, letting $L \in \{0,1\}^{n \times m}$ be the matrix corresponding to $\mathcal{S}$, $|\{(v_i,v_j) \in E: L(i) = L(j)\}| < \epsilon n^2$. It is strongly separating if $|\{(v_i,v_j) \in E: L(i), L(j)\text{ are not non-dominating}\}| < \epsilon n^2$.
\end{definition}

For  $\epsilon$-strongly separating set systems, when the number of  interventions is large, specifically $m\geq {1}/{\epsilon}$,  a simple approach is to partition the nodes into $1/\eps$ groups of size $\epsilon \cdot n$. We then assign the same weight $1$ vector to nodes in the same group and different weight 1 vectors to nodes in different groups. For $\epsilon$-separating set system, we first find an approximate minimum vertex cover, and then apply the above partitioning. In Appendix~\ref{app:ss}, we show that we get within a $2$ factor of the optimal (strongly) separating set system. Therefore, for the remainder of this paper we assume $m < 1/\epsilon$. While $m$ is an input parameter, smaller $m$ corresponds to fewer interventions and this is the more interesting regime. 

\section{Observable Graph Recovery}\label{sec:graph}

We start by considering the setting where we are given all edges in the observable graph $G$ (i.e., all direct causal relations between observable variables) e.g., by a domain expert, and wish to identify the direction of these edges. 
It is known that, assuming causal sufficiency (no latents), a separating set system is necessary and sufficient to learn $G$ \citep{eberhardt2007causation}. In Appendix \ref{app:graph} we show that this is also the case in the \emph{presence of latents} when we are given the edges in $G$ but not their directions. We also show that an $\epsilon$-separating set system is sufficient to approximately learn $G$ in this setting:
%
%

\begin{claim}\label{cl:epsSS}
Under the assumptions of Section \ref{sec:prelim}, if $\mathcal{S} = \{ S_1, S_2, \cdots S_m \}$ is an $\epsilon$-separating set system for $G$, $\mathcal{S}$ suffices to $\epsilon$-approximately learn $G$. 
\end{claim}
In particular, if $\mathcal{S}$ is an $\epsilon$-separating set system, we can learn all edges in $G$ that are separated by $\mathcal{S}$ up to $\epsilon n^2$ edges which are not separated.
%
%
Given Claim \ref{cl:epsSS}, our goal becomes to find an $\epsilon$-separating matrix $L_\epsilon$ for $G$ satisfying for some small approximation factor $\alpha$, $C(L_\epsilon) \le \alpha \cdot C(L^*)$ where $L^*$ is the minimum cost separating matrix for $G$. Missing technical details of this section are collected in Appendix~\ref{app:graph}. 

We follow the approach of ~\cite{neurips18}, observing that every node in an independent set of $G$ can be assigned the same vector in a valid separating matrix. They show that if we greedily peel off maximum independent sets from $G$ and assign them the lowest remaining weight vector in $\{0,1\}^m$ not already assigned as a row in $L$, we will find a $2$-approximate separating matrix. Their work focuses on chordal graphs where an MIS can be found efficiently in each step. However for general graphs $G$, finding an MIS (even approximately) is hard (see Appendix~\ref{app:disc}). Thus, in Algorithm~\ref{alg:ss_matrix}, we modify the greedy approach and  in each iteration we find a \textit{near} independent set with cost at least as large as the true MIS in $G$  (Def. \ref{def:nearMIS}).
Each such set has few internal edges, this leads to few non-separating assignments between edges of $G$ in $L_\epsilon$. 
Let $\epsilon$ be parameter that bounds the number of non-separating edges, and $\delta$ is the failure probability parameter of our Algorithm~\ref{alg:ss_matrix}. All the error parameters  are scaled appropriately (See Appendix~\ref{app:graph} for more details) when we pass them along in a procedure call to~\NearMIS~(line 5 in Algorithm~\ref{alg:ss_matrix}). 

\begin{algorithm}[H]
\caption{$\epsilon$-\textsc{Separating Matrix}$(G,m, \epsilon, \delta)$}
\label{alg:ss_matrix}
\begin{algorithmic}[1]
\State \textbf{Input} : Graph $G = (V, E)$, cost function $C:V \rightarrow \R^+$, $m$, error $\epsilon$, and failure probability $\delta$.
\State \textbf{Output} : $\epsilon$-Separating Matrix $L_{\epsilon} \in \{0,1\}^{n \times m}$.
\State Mark all vectors in $\{0,1\}^{m}$ as available. 
\While{$|V| > 0 $}
\State $S \leftarrow \NearMIS~(G, \epsilon ^2, \epsilon \delta )$
\State $\forall \ v_j \in S$, Set $L_{\epsilon}(j) $ to smallest weight vector available from $\{ 0, 1 \}^m$ and mark it unavailable.
\State Update $G$ by $E \leftarrow E \setminus E[S]$ and $V \leftarrow V \setminus S$.
\EndWhile
\State \return $\ L_\epsilon$
\end{algorithmic}
\end{algorithm}
Observe that any subset of fewer than $\epsilon n$ nodes has at most $\epsilon^2 n^2$ internal edges and so the $\NearMIS~(G, \epsilon ^2, \epsilon\delta)$ routine employed in Algorithm \ref{alg:ss_matrix} always returns at least $\epsilon n$ nodes. Thus the algorithm terminates in $1/\epsilon$ iterations. Across all $1/\epsilon$ \NearMIS's there are at most $\epsilon^2 n^2 \cdot 1/\epsilon = \epsilon n^2$ edges with endpoints assigned the same vector in $L_\epsilon$, ensuring that $L_\epsilon$ is indeed $\epsilon$-separating for $G$.

In Algorithm~\ref{alg:near_mis}, we implement the \NearMIS~routine by using the notion of a $(\rho, \gamma, \epsilon)$-Independent-Set (Definition~\ref{def:rhoepsis}). 
We  find a value of $\rho$ that achieves close to the MIS cost via a search over decreasing powers of ${(1+\gamma)}$.  In Algorithm~\ref{alg:mis_tester} we show how to obtain a $(\rho, \gamma, \epsilon)$-Independent-Set (denoted by $S$) whenever the cost of MIS in $G$ is at least $\rho \cdot C(V)$. So, $C(S) \geq \rho C(V) - \rho \gamma C(V)$ and we might lose a cost of at most $\gamma \rho C(V)$ compared to the MIS cost. Therefore, we add $\eps \cdot n$ nodes of highest cost (denoted by $S_{\eps/2}$) to $S$ and argue that by setting $\gamma = O(\eps/W)$, $S \cup S_{\eps/2}$ has a cost at least the cost of MIS, i.e., $S \cup S_{\eps/2}$ is a $(0,\eps)$-\NearMIS.
\begin{algorithm}[h]
\caption{\NearMIS}
\label{alg:near_mis}
\begin{algorithmic}[1]
\State \textbf{Input :} Graph $G(V, E)$, cost function $C:V \rightarrow \R^+$, error $\epsilon$, and failure probability $\delta$. 
\State \textbf{Output :} Set of nodes that is a $(0, \epsilon)$-\NearMIS~ in $G$.
\State Initialize $\rho = 1$, and let $T$ be the set of $\sqrt{\epsilon} n$ nodes in $G$ with the highest cost.
\While{$\rho \geq {\sqrt{\epsilon}}$}
\State $S \leftarrow \IndepSet(G, \rho, \eps/8 W, \eps, \delta')$ where $\delta' = \eps\delta/4 W \log(1/\epsilon)$ 
\State Let $S_{\eps/2}$ denote the highest cost ${\eps \cdot n}/{2}$ nodes in $V \setminus S$.
\If{$C(S \cup S_{\eps/2}) \geq C(T) \textbf{ and }|E[S \cup S_{\eps/2}]| \leq \eps n^2$}
\State \return~$S \cup S_{\eps/2}$
\EndIf
\State $\rho = {\rho}/{(1+\gamma)}$    
\EndWhile
\State \return~$T$
\end{algorithmic}
\end{algorithm}

\subsection{$(\rho, \gamma, \epsilon)-\IndepSet$}
In this section, we introduce several new ideas and build upon the results for finding a $(\rho, 0, \epsilon)$-Independent-Set which has been used to obtain independent set property testers for graphs with unit vertex costs~\citep{GGR98}. First, we describe an overview of the general approach.

\vspace*{1ex}

\noindent \textbf{Unit Cost Setting}. Suppose $S$ is a fixed MIS in $G$ with $|S| \ge \rho \cdot n$ and $U \subset S$. Let $\Gamma(u)$ represent the set of nodes that are neighbors of node $u$ in $G$. \text{Let } $$\Gamma(U) = \bigcup_{u \in U} \Gamma(u) \text{ and }\widebar{\Gamma}(U) = V\setminus \Gamma(U).$$ Here, $\widebar{\Gamma}(U)$ denotes the set of nodes with no edges to any node of $U$. We claim that $ S \subseteq \widebar{\Gamma}(U)$. First, we observe that $S \subseteq \widebar{\Gamma}(S)$ as $S$ is an independent set so no node in $S$ is a neighbor of another node in $S$ (i.e., all nodes in $S$ are in $\widebar{\Gamma}(S)$). Then, we use the fact $\widebar{\Gamma}(S) \subseteq \widebar{\Gamma}(U)$ since $U \subseteq S$ to conclude $S \subseteq \widebar{\Gamma}(U)$.
Further, \cite{GGR98} proves that, if $U$ is sampled randomly from $S$, taking the lowest degree $\rho \cdot n$ nodes in the induced subgraph on $\widebar{\Gamma}(U)$ will with high probability yield a $(0,\eps)$-\NearMIS~for $G$. Intuitively, the nodes in $\widebar{\Gamma}(U)$ have no connections to $U$ and thus are unlikely to have many connections to $S$. 

To find a $U$ that is fully contained in $S$, we can sample a small set of nodes in $G$; since we have $|S| \ge \rho \cdot n$ the sample will contain with good probability a representative proportion of nodes in $S$. We can then brute force search over all subsets of this sampled set until we hit $U$ which is entirely contained in $S$ and for which our procedure on $\widebar{\Gamma}(U)$ returns a $(\rho, 0, \eps)$-Independent-Set, i.e., a \NearMIS.

\vspace*{1ex}
\noindent \textbf{General Cost Setting}. In the general cost setting, when $S$ is a high cost MIS, may not contain a large number of nodes, making it more difficult to identify via sampling. To handle this, we partition the nodes based on their costs in powers of $(1+\gamma)$ into $k = O(\gamma^{-1} \log W)$ (where $W$ is the maximum cost of a node in $V$) partitions $V_1,\ldots,V_k$. 

A \emph{good partition} is one that contains a large fraction of nodes in $S$: at least $\gamma \rho |V_i|$. Focusing on these partitions suffices to recover an approximation to $S$. Intuitively, all bad partitions have few nodes in $S$ and thus ignoring nodes in them will not significantly affect the MIS cost.

\begin{definition}[$(\gamma, \rho)$-good partition] Let $S$ be an independent set in $G$ with cost $\ge \rho C(V)$. Then $F_{(\gamma,\rho)} = \{ i \mid |V_i \cap S | \geq \gamma \rho |V_i| \}$ is the set of \emph{good partitions} of $V$ with respect to $S$.\label{def:good_partition}
\end{definition}

\begin{claim}
Suppose $S$ is an independent set in $G$ with cost $C(S) \geq \rho C(V)$, then, there exists an independent set $S' \subseteq S$ such that $C(S') \geq \rho(1-2\gamma) C(V)$ and $S' \cap V_i = S \cap V_i$ for all $i \in F_{(\gamma,\rho)}$.
\label{cl:approx_WIS}
\end{claim}

While we do not a priori know the set of good partitions, if we sample a small number $t$ of nodes uniformly from each partition, with good probability, for each good partition we will sample $\gamma \rho t/2$ nodes in $S$. We search over all possible subsets of partitions and in one iteration of our search, we have all the good partitions denoted by $\{ V_1, V_2 \cdots V_\tau \}$. Now, for such a collection of good partitions, we search over all possible subsets 
$\mathcal{U} = U_1 \cup U_2 \cdots \cup U_\tau$ where $|U_i| = \gamma \rho t/2$ and in at least one instance have all $U_i$ in good partitions fully contained in $S$. Let 
$$Z(\mathcal{U}) := \bigcup^\tau_{i=1} V_i \setminus \bigcup^\tau_{i=1} {\Gamma}(U_i)$$
be the nodes in every \emph{good} partition $V_i$ with no connections to any of the nodes in $U_i$. Analogous to unit cost case, we sort the nodes in a  \emph{good} partition $V_i$ by their degree in the induced subgraph on $Z(\mathcal{U})$. We select low degree nodes from each partition until the sum of the total degrees of the nodes selected is $\epsilon n^2/k$. We output union of all such nodes iff it is a $(\rho, 3\gamma, \epsilon)$-independent set. One key difference is that while including nodes from $Z(\mathcal{U})$,
we do not include the nodes in the sorted order until sum of degrees is $\eps n^2$. Instead, we process each \emph{good} partition and include the nodes from each partition separately. Later, we will argue that by doing so we have made sure that the cost contribution of a particular partition is accounted for accurately.

\begin{algorithm}[h]
\caption{$(\rho, \gamma, \eps)$ \IndepSet}
\label{alg:mis_tester}
\begin{algorithmic}[1]
\State \textbf{Input} : Graph $G = (V, E)$, cost function $C:V \rightarrow \R^+$, parameters $\rho, \gamma, \epsilon$ and $\delta$
\State \textbf{Output} : $(\rho, 3\gamma, \epsilon)$ independent set~in $G$ if one exists.
\State For $i=1, \ldots, k,$ define $V_i = \{ v\in V \mid (1+\gamma)^{i-1} \leq C(v) < (1+\gamma)^i \}$ where $k = \gamma^{-1} \log W$ 
\State 
Sample $t = O(\frac{k \log(k/\epsilon \delta)}{\epsilon \gamma \rho})$ nodes $\tilde{V}_i$ in each partition $V_i$.

\For{every collection of partitions $\{ {V}_1, {V}_2, \cdots {V}_\tau \} \subseteq \{ {V}_1, {V}_2, \cdots {V}_k \}$}
\For{ $\mathcal{U} = U_1 \cup U_2 \cup \cdots \cup U_\tau$ such that $U_i \subseteq \tilde{V}_i$ with size $\gamma \rho t/2$ for all $i \in [\tau]$}
\State Let $Z(\mathcal{U}) :=  \bigcup_{i=1}^\tau V_i \setminus \bigcup_{i=1}^\tau  \Gamma(U_i)$.
\For{$i = 1\ldots \tau$}
\State Sort nodes in $Z(\mathcal{U}) \cap V_i$ in increasing order of degree in the induced graph on $Z(\mathcal{U})$. 
\State Let $\hat Z_i(\mathcal{U}) \subseteq 
 Z(\mathcal{U}) \cap V_i$ be set of nodes obtained by considering the nodes in the sorted order until the total degree is $\epsilon n^2/k$. 
\EndFor
\State Let $\hat{Z}(\mathcal{U}) = \bigcup_{i=1}^\tau {\hat Z_i(\mathcal{U})}$.
\State \return~$\hat{Z}(\mathcal{U})$ if $C(\hat{Z}(\mathcal{U})) \geq \rho(1-3\gamma) C(V)$.
\EndFor
\EndFor
\end{algorithmic}
\end{algorithm}
\setlength{\textfloatsep}{4pt}

By construction, our output, denoted by $\hat Z(\mathcal{U})$ will have at most $\epsilon n^2$ internal edges. Thus, the challenge lies in analyzing its cost. We argue that in at least one iteration, all chosen $U_i$ for good partitions will not only lie within the MIS $S$, but their union will accurately represent connectivity to  $S$.  Specifically, any vertex $v \in Z(\mathcal{U})$, i.e., with no edges to $U_i$ for all $i \in F_{(\gamma,\rho)}$, should have few edges to $S$. We formalize this notion using the definition of $\epsilon_2$-IS representative subset below.

\begin{definition}[$\epsilon_2$-IS representative subset]
$R \subseteq \bigcup_{i \in F_{(\gamma,\rho)} } \left( S \cap V_i \right)$
is an $\epsilon_2$-IS representative subset of $S$ if for all but $\epsilon_2 n$ nodes of \emph{good} partitions i.e., $\bigcup_{i \in F_{(\gamma, \rho)}}V_i$, we have the following property:
 $$ \text{Suppose }v \in \bigcup_{i \in F_{(\gamma, \rho)}}V_i : \ \ \text{ if } \Gamma(v) \cap R = \emptyset  \text{ then } |\Gamma(v) \cap S| \leq \epsilon_2 n.$$ 
\label{def:rep_subset}
\end{definition}

We show that there is a \emph{$\epsilon_2$-IS representative subset} containing at least $\gamma \rho t/2$ nodes from each good partition among our sampled nodes $\bigcup_{i=1}^k \widetilde{V_i}$. Setting $\epsilon_2 = \epsilon/2k$ we have:
\begin{lemma}
If  $t = O(\frac{k \log(k/\epsilon \delta)}{\epsilon \gamma \rho})$ nodes are uniformly sampled from each partition $V_i$ to give $\tilde{V}_i$, with probability $1-\delta$, there exists an $\epsilon/2k$-IS representative subset $R$ such that, for every $i \in F_{(\gamma,\rho)} $, $|\tilde{V}_i \cap R| = \gamma \rho t/2$.
\label{lem:eps_representative}
\end{lemma}
Lemma \ref{lem:eps_representative} implies that in at least one iteration, our guess $\mathcal{U}$ restricted to the \emph{good partitions} is in fact an $\epsilon/2k$-IS representative subset. Thus, nearly all nodes in $Z(\mathcal U)$ lying in good partitions have at most $\epsilon n/2k$ edges to $S$.

In the graph induced by nodes of $Z(\mathcal{U})$, with edge set $E[Z(\mathcal{U})]$, consider the degree incident on nodes of $S \cap V_i$ for each partition $V_i$. As there are at most $n$ nodes in $V_i$, from Defn.~\ref{def:rep_subset}, we have the total degree incident on $S \cap V_i$ is at most $ \epsilon n^2/k$. Thus, including the nodes with lowest degrees in $\hat Z_i(\mathcal{U})$ until the total degree is $\epsilon n^2/k$ will yield a set of nodes at least as large as $S \cap V_i$. Since all nodes in $V_i$ have cost within a $1\pm \gamma$ factor of each other, we will have $C(\hat Z_i(\mathcal{U})) \ge (1-\gamma)\cdot C(S \cap V_i)$. As the cost of $S$ in the bad partitions is small, using Claim~\ref{cl:approx_WIS}, we have $\hat{Z}(\mathcal{U}) = \bigcup_{i=1}^\tau {\hat Z_i(\mathcal{U})}$ is a $(\rho, O(\gamma),\epsilon)$-independent set.

\subsection{Approximation Guarantee}
Overall, Algorithm \ref{alg:mis_tester} implements a $(\rho, \gamma, \eps)-$\IndepSet~as required by Algorithm \ref{alg:near_mis} to compute a \NearMIS~in each iteration of Algorithm \ref{alg:ss_matrix}. It just remains to show that, by greedily peeling off \NearMIS~from $G$ iteratively,  Algorithm~\ref{alg:ss_matrix} achieves a good approximation guarantee for $\epsilon$-Approximate Learning $G$. To do this, we use the analysis of a previous work from \cite{neurips18}. In their work, an exact MIS is computed at each step, since their graph is chordal so the MIS problem is polynomial time solvable~\citep{neurips18}. However, the analysis extends to the case when the set returned has cost that is at least the cost of MIS (in our case a \NearMIS), allowing us to achieve near $2$-factor approximation, as achieved in \citep{neurips18}.  Our final result is:

\begin{theorem}\label{thm:main_ss}
For any $m \geq \eta \log 1/\epsilon$ for some constant $\eta$, with probability $\ge 1-\delta$, Algorithm~\ref{alg:ss_matrix} returns $L_\epsilon$ with 
$C(L_\epsilon) \leq (2  + \exp{(-\Omega(m))})  \cdot C(L^*)$, where $L^*$ is the min-cost separating matrix for $G$. Moreover $L_\epsilon$ $\epsilon$-separates $G$. Algorithm~\ref{alg:ss_matrix} has a running time $O(n^2 f(W, \epsilon, \delta))$ where 
$f(W, \epsilon, \delta) = O\left( \frac{W}{\eps^2} \log \frac{1}{\eps} \exp{ \left(O\left(    \frac{W^2 \log^2 W}{\eps^6} \log \frac{W}{\epsilon} \log \frac{W \log W \log 1/\eps}{{\epsilon \delta}} \right)\right) } \right).$
\end{theorem}

\section{Ancestral Graph Recovery}\label{sec:ancG}
In Section \ref{sec:graph}, we assumed knowledge of the edges in the observable graph $G$ and sought to identify their directions. 
In this section, we relax the assumption, assuming we are given any undirected supergraph $H$ of $G$ i.e., it includes all edges of $G$ and may also include edges which do not represent causal edges. When given such a graph $H$, we cannot recover $G$ itself and therefore, we seek to recover all directed edges of the ancestral graph $\Anc(G)$ appearing in $H$ (i.e., the set of intersecting edges), which we denote by $\Anc(G) \cap H$. This problems strictly generalizes that of Section \ref{sec:graph}, as when $H = G$ we have $\Anc(G) \cap H = G$.  Missing details of this section are collected in Appendix~\ref{app:ancG}.

First, we show that to recover $\Anc(G) \cap H$, a strongly separating system (Def \ref{def:sss}) for $H$ is both necessary and sufficient. Furthermore, an $\epsilon$-strongly separating system suffices for approximate learning. We formalize this using the following lemma:

\begin{lemma}\label{clm:ssssSuffices}
Under the assumptions of Section \ref{sec:prelim}, if $\mathcal{S} = \{ S_1, S_2, \cdots S_m \}$ is an $\epsilon$-strongly separating set system for $H$ , $\mathcal{S}$ suffices to $\epsilon$-approximately learn $\Anc(G) \cap H$.
\end{lemma}

Given Lemma \ref{clm:ssssSuffices},  our goal becomes to find an $\epsilon$-strongly separating matrix for $H$, $L_\epsilon$ with cost within an $\alpha$ factor of the optimal strongly separating matrix for $H$, for some small $\alpha$.  To do so, our algorithm builds on the separating set system algorithm of Section \ref{sec:graph}. We first run Algorithm~\ref{alg:ss_matrix} to obtain an $\epsilon$-separating matrix $L^S_{\epsilon}$ and construct $S_1, S_2, \cdots S_{1/\epsilon}$ where each set $S_i$ contains all nodes assigned the same vector in $L^S_{\epsilon}$ -- i.e., $S_i$ corresponds to the \NearMIS~computed at step $i$ of Algorithm \ref{alg:ss_matrix}.  We form a new graph by contracting all nodes in each $S_i$ into a single \textit{super node} and denote the resulting at most $1/\epsilon$ vertices by $V_S$. In~\cite{AKMM2020}, the authors give a $2$-approximation algorithm for finding a strongly separating matrix on a set of nodes, provided the graph on these nodes is complete. 
 As $H$ is an arbitrary super graph of $G$, the contracted graph on $V_S$ is also arbitrary. However we simply assume the worst case, and run the Algorithm of~\cite{AKMM2020} on it to produce $L^{SS}_{\epsilon}$, which strongly separates the complete graph on $V_S$. It is easy to show that as a consequence, $L^{SS}_\epsilon$ $\epsilon$-strongly separates $H$.%

\begin{algorithm}[ht]
\caption{\textsc{Ancestral Graph}$(H, m, \epsilon,\delta)$}
\label{alg:sss_matrix}
\begin{algorithmic}[1]
\State $L^S_{\epsilon} := \epsilon$-\textsc{Separating Matrix}$(H,m, \epsilon, \delta)$.
\State Construct $S_1, S_2, \cdots S_{1/\epsilon}$ where each set $S_i$ contains nodes assigned the same vectors in $L^S_{\epsilon}$.
\State Construct a set of nodes $V_S$ by representing $S_i$ as a single node $w_i \text{ and }C(w_i) = \sum_{u \in S_i} C(u)$.
\State $L^{SS}_{\epsilon} :=$\textsc{ SSMatrix}$(V_S, m)$ from \cite{AKMM2020}.
\State \return~$L^{SS}_{\epsilon}$
\end{algorithmic}
\end{algorithm}

To prove the approximation bound, we extend the result of~\cite{AKMM2020}, showing that their algorithm actually achieves a cost at most $2$ times the cost of a \emph{separating matrix} for the complete graph on $V_S$ which satisfies two additional restrictions: (1) it does not assign the all zeros vector to any node and (2) it assigns the same number of weight one vectors as the optimal strongly separating matrix. 
Further, we show via a similar analysis to Theorem~\ref{thm:main_ss} that this cost on $V_S$ is bounded by $2$ times the cost of the optimal strongly separating matrix on the contracted graph over $V_S$.  Combining these bounds yields the final $4$ approximation guarantee of Theorem~\ref{thm:main_sss}.

\begin{theorem} 
Let $m \geq \eta \log 1/\epsilon$ for some constant $\eta$ and $L^{SS}_{\epsilon}$ be matrix returned by Algorithm~\ref{alg:sss_matrix}. Then with probability $\ge 1-\delta$, $L^{SS}_{\epsilon}$ is an $\epsilon$-strongly separating matrix for $H$ and $C(L_\epsilon^{SS}) \le (4+\exp{(-\Omega(m))})   \cdot C(L^*)$
where $L^*$ is the min-cost strongly separating matrix for $H$. Algorithm~\ref{alg:sss_matrix} runs in time 
$O(n^2 f(W, \epsilon, \delta))$ where $$f(W, \epsilon, \delta) = O\left( \frac{W}{\eps^2} \log \frac{1}{\eps} \exp{ \left(O\left(    \frac{W^2 \log^2 W}{\eps^6} \log \frac{W}{\epsilon} \log \frac{W \log W \log 1/\eps}{{\epsilon \delta}} \right)\right) } \right).$$
\label{thm:main_sss}
\end{theorem}

\section{Hyperfinite Graphs : Better Guarantees}
In this section, we show that when $G$ has maximum degree $\Delta$ and satisfies hyperfinite property, we can obtain the same approximation guarantees, but the number of edges that are not (strongly) separated can be improved to $\eps \cdot n \cdot \Delta$. Informally, a hyperfinite graph can be partitoned into small connected components by removing $\eps \cdot n$ edges for every $\epsilon > 0$. Bounded degree hyperfinite graphs include the class of bounded-degree graphs with excluded minor~\citep{alon1990separator}, such as planar graphs, constant tree-width graphs, and also non-expanding graphs~\citep{czumaj2009testing}. We defer the full details of the results to Appendix~\ref{app:hyperfinite}.

\section{Open Questions}
We highlight that in both the settings, although we consider the presence of latents in the system, in this paper, we provide results for learning causal relations among only the observable variables. Identification of latents is an important goal and has been well-studied~\citep{neurips17,kocaoglu2019characterization, AKMM2020} when the objective is to minimize the \emph{number of interventions}. However, in~\cite{AKMM2020}, for the \textit{linear cost model}, the authors argue that there is no good cost lower bound known, even for recovering the observable (rather than ancestral) graph in the presence of latents. This makes the development of algorithms with approximation guarantees in terms of the optimum cost difficult. We view addressing this difficulty as a major open question. 

{ Our results on bounded degree graphs make an additional assumption that the graph is hyperfinite, which gives more structure but still captures many graph families. It is an interesting open question if we can extend them to general sparse graphs. This setting is challenging since even finding a Near-MIS with $\epsilon \cdot |E|$ edges is still open and likely to be hard~\citep{ron2010algorithmic}. We conjecture that if $|E| = O(n)$, a constant approximation for our objectives is not possible (assuming standard complexity theoretic conjectures) if we must separate all but $\epsilon \cdot |E|$ many edges.}

It would also be very interesting to extend our work to the setting where we seek to identify a specific subset of edges of the causal graph, or where certain edges are `more important' than others. We hope that our work is a first step in this direction, introducing the idea of partial recovery to overcome hardness results that rule out non-trivial approximation bounds for full graph recovery in the linear cost model.
\pagebreak

\bibliographystyle{plainnat}
\bibliography{references}
\pagebreak
\appendix
\section{Discussions}\label{app:disc}
\subsection{Semi-Markovian Assumption}
Our assumption that each latent only affects two observable variables is commonly known as the semi-Markovian condition and is standard in the literature, e.g., see~\cite{tian2003identification, neurips17}. In fact, using (pairwise) conditional independence tests, it is impossible to discover latent variables that affect more than two observables, even with unlimited interventions. Consider observables $x, y, z$ and a latent $l_{xyz}$ that is a parent of all them. If we test whether $x, y$, and $z$ are all pairwise independent and they all turn out to be false, we can’t distinguish the cases where a single latent $l_{xyz}$ or three separate latents $l_{xy}, l_{yz}$ and $l_{xz}$ are causing this non-independence. Thus, we cannot remove the assumption without changing our intervention model or making more restrictive assumptions. As an example, ~\cite{silva2006learning} considers the case when latents affect more than two observables, however, they make very strong assumptions -- that there are no edges between observables, and each observable has only one latent parent.

\subsection{Hardness of Independent Set}

For the linear cost model, the problem of learning a causal graph was introduced in ~\cite{icml17}. It was shown recently that the problem of obtaining an optimum cost set of interventions is NP-hard \citep{neurips18}. Under causal sufficiency (no latents), it is well known that the undirected graph (also called Essential Graph~\citep{zhang2008completeness, neurips18}) recovered after running the $IC^*$ algorithm is chordal. Further, an intervention set which is a separating set system (Def. \ref{def:sss}) for the Essential Graph of $G$  is both necessary and sufficient~\citep{eberhardt2007causation, shanmugam2015learning} for learning the causal graph. 

The authors of \cite{neurips18} give a greedy algorithm to construct a $2$-approximation to the optimal cost separating set system of the essential graph. Their algorithms requires at each step finding a maximum independent set in $G$ and peeling it off the graph, and is the basis for our approach in Section \ref{sec:graph}. Since $G$ is  \textit{chordal}
, there is an algorithm for finding an exact maximum independent set in polynomial time~\citep{frank1975some}. However, without the assumption of causal sufficiency, we cannot directly extend their algorithm, since  finding a maximum independent set in a general graph $G$ is NP-hard~\citep{cormen2009introduction}. Moreover, finding an approximate independent set within a factor of $n^{\epsilon}$ for any $\epsilon > 0$ in polynomial time is also not possible unless $NP \subseteq BPP$~\citep{feige1996interactive}.

\section{Missing Details From Section~\ref{sec:ss}}\label{app:ss}
\subsection{$2\log n$ Approximation Algorithm for Separating Set System}
In this section, we show that the algorithm presented in section~\ref{sec:ss} obtains a $2\log n$-optimal separating set system for a given graph $G$. To do so, we first make the following two simple claims. 
Let $\mathcal{S}^* = \{ S_1, S_2, \cdots S_m \}$ be the minimum cost separating set system for $G$ and $I$ denote the maximum cost independent set in $G$.

\begin{definition}{(Vertex Cover).} A set of nodes $S$ is a vertex cover for the graph $G(V, E)$, if for every edge $(u, v) \in E$, we have $\{ u, v\} \cap S \neq \phi$.
\label{def:vc}
\end{definition}

\begin{claim}\label{cl:vc}
The set of vertices in $V \setminus I$ forms a minimum weighted vertex cover for $G$.
\end{claim}
\begin{proof}
Suppose $X$ denote a minimum weighted vertex cover in $G$, then, $V \setminus X$ is an independent set in $G$. We have $C(X) = C(V) - C(V\setminus X) \geq C(V) - C(I)$ as $I$ is maximum cost independent set. Observe that the vertex cover given by $X := V\setminus I$ satisfies the above equation with equality. Hence, the claim.
\end{proof}
\begin{claim}\label{cl:is}
$C(\mathcal{S}^*) \geq C(V \setminus I)$
\end{claim}

\begin{proof}
Let $L^*$ denote optimal separating matrix corresponding to $\mathcal{S}^*$. We can rewrite $C(\mathcal{S}^*)$ in terms of $C(L^*) = \sum_{j=1}^n C(v_j)\norm{L(j)}_1$. It is easy to observe that every node in an independent set of $G$ can be assigned the same vector in a separating matrix. So, nodes with weight zero in $L^*$ are from an independent set (say $I_{L^*}$) in $G$. As weight of the vectors assigned to remaining nodes in $L^*$ is at least $1$, we have $C(L^*) \geq C(V) - C(I_{L^*}) \geq C(V) - C(I)$, using the definition of $I$. 
\end{proof}

Combining Claims~\ref{cl:vc} and ~\ref{cl:is}, we can observe that a good approximation for weighted vertex cover will result in a good approximation for separating set system. There is a well known $2$-approximation algorithm for weighted vertex cover problem using linear programming that runs in polynomial time (Page 10, Theorem 1.6~\cite{williamson2011design}). 


\begin{lemma}
If $m \geq 2\log n$, then, there is an algorithm that returns a separating set system that is $2\log n$-optimal. 
\end{lemma}
\begin{proof}
Let $X$ denote the minimum weighted vertex cover which is a $2$-approximation obtained using the well known linear programming relaxation~\citep{williamson2011design}. In our algorithm, we assign every node in $X$ with a unique vector of weight $\log n$. This is feasible because the set of nodes in $V\setminus X$ form an independent set, and ${m \choose \log n} \geq {2\log n\choose \log n} \geq n$. Combining Claims~\ref{cl:vc} and ~\ref{cl:is}, we have
$$C(L) = \log n \  C(X) \leq 2 \log n \ C(V \setminus I) \leq 2\log n \  C(\mathcal{S}^*). $$
\end{proof}

\subsection{Algorithms for $\epsilon$-(Strongly) Separating Set System when $m \geq {1}/{\epsilon}$}

\noindent \textbf{$\epsilon$-Separating Set System}.
For $\epsilon$-separating set system on $G(V, E)$, we first find a $2$-approximate minimum weighted vertex cover $X$ using the well-known linear programming based algorithm from~\cite{williamson2011design} (Refer Page 10, Theorem 1.6 in~\cite{williamson2011design}). We then partition the nodes of  $X$ randomly into $1/\eps$ groups of expected size $\eps \cdot n$. We then assign the same weight $1$ vector to nodes in the same group and different weight $1$ vectors to nodes in different groups. This is possible since $m \ge 1/\epsilon$. It is easy to see that the total number of edges that are not separated  on expectation is $\epsilon |E| \leq \eps n^2$. For the remaining nodes in $V \setminus X$ that form an independent set, we assign the zero vector. Therefore, total cost of $\eps$-separating set system is given by $C(X)$. From Claim~\ref{cl:vc}, we have $C(X) \leq 2 C(V \setminus I)$ where $I$ is maximum weighted independent set in $G$. Using Claim~\ref{cl:is}, we have $C(X) \leq 2 C(\mathcal{S}^*)$ where $\mathcal{S}^*$ is optimal separating set system for $G$. Therefore, we get within a $2$ factor of the optimal separating set system.

\vspace*{1ex}
\noindent \textbf{$\epsilon$-Strongly Separating Set System}.
For $\epsilon$-strongly separating set system on $H(V, E)$, we partition the nodes randomly into $1/\eps$ groups of expected size $\epsilon \cdot n$. We then assign the same weight $1$ vector to nodes in the same group and different weight $1$ vectors to nodes in different groups. This is possible since $m \ge 1/\epsilon$. It is easy to see that the total number of edges that are not strongly separated on expectation is $\eps |E| \leq \eps n^2 $. As every vector assigned to a node in a valid strongly separating matrix should have weight at least $1$, this results in an $\epsilon$-strongly separating matrix, and the corresponding set system with optimal cost.

\section{Missing Details From Section~\ref{sec:graph}}\label{app:graph}

In this section, we refer to the conditional independence test described in section~\ref{sec:prelim} as CI-test.
\begin{claim}\label{cl:SS_necessary_appendix}
Suppose a set on interventions $\mathcal{S} = \{ S_1, S_2, \cdots S_m \}$ is used for learning the edges of an undirected causal graph $G$. Then, under the assumptions of section \ref{sec:prelim}, $\mathcal{S}$ is a separating set system for $G$.
\end{claim}
\begin{proof}
First, we show that when $\mathcal{S}$ is a separating set system for $G$, we can recover the directions of $G$. Consider an edge $(v_i, v_j) \in G$ and let $S_k \in \mathcal{S}$ be such that $v_i \in S_k $ and $v_j \not \in S_k$. As $\mathcal{S}$ is a separating set system, we know that such a set $S_k$ exists for every edge in $G$. Consider the CI-test between $v_i$ and $v_j$ in the interventional distribution $\doo(S_k)$. If the test returns that $v_i  \indep v_j \mid \doo(S_k)$, then, we infer $v_i \rightarrow v_j$, otherwise we infer that $v_i \leftarrow v_j$. When we intervene on $v_i$ obtained by $\doo(S_k)$, the latent edges affecting $v_i$ and all other incoming edges to $v_i$ are removed. As we know that there is a causal edge between the two variables, if the independence test returns true, it must mean that there is no incoming edge into $v_i$ from $v_j$.

In~\cite{eberhardt2007causation}, it was shown that a separating set system is necessary for learning the directions among the observable variables assuming causal sufficiency. As we are trying to recover $G$ using interventions, such a condition will also hold for our case that is a generalization when not assuming causal sufficiency.  Hence, the claim. 

\end{proof}

\begin{claim}(Claim \ref{cl:epsSS} restated)
Under the assumptions of Section \ref{sec:prelim}, if $\mathcal{S} = \{ S_1, S_2, \cdots S_m \}$ is an $\epsilon$-separating set system for $G$ , $\mathcal{S}$ suffices to $\epsilon$-approximately learn $G$. 
\end{claim}
\begin{proof}
Given $\mathcal{S}$ denotes an $\epsilon$-separating set system for $G(V, E)$.
So, there are at most $\epsilon n^2$ edges $(u, v) \in E$ such that for all $i \in [m]$, either $\{u, v\} \cap S_i = \phi$ or $\{u, v\} \cap S_i = \{u, v\}$. For every such edge, any intervention on a set in $\mathcal{S}$, say $S_i$ cannot recover the direction from a CI-test $u \indep v \mid \doo(S_i) ?$ because for both the cases $u \leftarrow v$ or $u \rightarrow v$, the CI-test returns that they are dependent. For the remaining edges $(u, v) \in E$, in the intervention $S_j$ where $\{ u, v\} \cap S_j  = \{ u \}$, we can recover the direction using the CI-test : $u \rightarrow v$ if $u \notindep v \mid \doo(S_j)$ and $u \leftarrow v$ otherwise. From Def.~\ref{def:eps_learn}, we have that $\mathcal{S}$ $\epsilon$-approximately learns $G$. 
\end{proof}

\begin{claim}(Claim~\ref{cl:approx_WIS} restated)
Suppose $S$ is an independent set in $G$ with cost $C(S) \geq \rho C(V)$, then, there exists an independent set $S' \subseteq S$ such that $C(S') \geq \rho(1-2\gamma) C(V)$ and $S' \cap V_i = S \cap V_i$ for all $i \in F_{(\gamma,\rho)}$.
\label{cl:approx_WIS_appendix}
\end{claim}
\begin{proof}
Construct $S'$ using $(\gamma, \rho)$-good partitions of $V$. For every $i \in F_{(\gamma, \rho)}$, include $S \cap V_i$ in $S'$. Therefore, we have 
\begin{align*}
    C(S') &= C(S) - \sum_{i \not \in F_{(\gamma, \rho)}} C(S \cap V_i) \\
          &\geq \rho C(V) - \gamma \rho \sum_{i \not \in F_{(\gamma, \rho)}} |V_i| (1+\gamma)^i \\
          &\geq \rho C(V) - \gamma \rho (1+\gamma)\sum_{i \not \in F_{(\gamma, \rho)}} |V_i| (1+\gamma)^{i-1}\\ 
          &\geq \rho C(V) - \gamma \rho (1+\gamma) C(V) \\
          &\geq \rho(1-2\gamma) C(V).
\end{align*}
\end{proof}

\begin{lemma}(Lemma~\ref{lem:eps_representative} restated)
If  $t = O(\frac{k}{\eps \gamma \rho} \log \frac{4k}{\eps \delta})$ nodes are uniformly sampled from each partition $V_i$ to give $\tilde{V}_i$, with probability $1-\delta$, there exists an $\epsilon/2k$-IS representative subset $R$ such that, for every $i \in F_{(\gamma,\rho)} $, $|\tilde{V}_i \cap R| = \gamma \rho t/2$.
\label{lem:eps_representative_appendix}
\end{lemma}
\begin{proof}

Consider a good partition $V_i$ for some $i \in F_{(\gamma,\rho)}$. So,  $|V_i \cap S| \geq \gamma \rho |V_i|$. As $\tilde{V}_i$ consists of $t$ nodes that are uniformly sampled from $V_i$, using Hoeffding's inequality (\cite{Bardenet_2015, hoeffding1994probability}), we know that $|\tilde{V}_i \cap S| \geq \gamma \rho t/2$ with {probability} at least 
$$ 1-\exp(-\gamma \rho t/8) \geq 1-\exp\left(-\frac{k}{\eps}\log \frac{4k}{\epsilon \delta}\right) \geq 1-\frac{\delta}{2k}.$$

Applying union bound, we have for every $i \in F_{(\gamma, \rho)}$, $|\tilde{V}_i \cap S| \geq \gamma \rho t/2$ with probability at least $$ 1-k \frac{\delta}{2k} \geq {1-\frac{\delta}{2}}.$$

Consider the union of all subsets  $U_j \subseteq \tilde{V}_j \cap S$ of good partitions such that $|U_j| = \gamma \rho t/2$, i.e., 
\[R = \bigcup\limits_{j=1 \ \mid \ j \in F_{(\gamma, \rho)}}^k U_j. \] 

We claim that $R$ is a $\eps/2k$-IS representative subset of $V$ by arguing that if $v$ has no neighbours in $R$, then, the degree to $S$ is more than $\epsilon n/2k$ with \textit{low} probability. 

First, consider the case when $v \in S$, then $|\Gamma(v) \cap S| = 0 \text{ and } \Gamma(v) \cap R = \phi$. Suppose  $v \in V_j\setminus S$ for some $j \in F_{(\gamma, \rho)}$ and $|\Gamma(v) \cap S| \geq \eps n/2k$. If $\Gamma(v) \cap R = \phi$, then $\Gamma(v) \cap R \cap V_i = \phi$ for all $i \in F_{(\gamma, \rho)}$. As $R$ is formed using the sampled nodes, we have that every node in $R$ should be from $V_i \setminus (\Gamma(v) \cap S \cap V_i)$ for the condition $\Gamma(v) \cap R = \phi$ to be satisfied. As every element in $R$ is chosen uniformly at random from the respective good partitions independently, we have :
\begin{align*}
\Pr_{\forall i, U_i \sim V_i}[ \forall i : \Gamma(v) \cap R \cap V_i = \phi  \textbf{ and } |\Gamma(v) \cap S| > \epsilon n/2k ] 
&\leq  \Pi_{i \in F_{(\gamma, \rho)}}\left(\frac{|V_i| - |\Gamma(v) \cap S \cap V_i|  }{|V_i|} \right)^{|U_i|}\\
&\leq \exp{\left(-\sum_i \frac{|U_i||\Gamma(v) \cap S \cap V_i|}{|V_i|} \right)}\\
&\leq \exp{\left(-\frac{\gamma \rho t}{n}\sum_i {|\Gamma(v) \cap S \cap V_i|} \right)}\\
&\leq \exp{\left( -\frac{\gamma \rho t}{n} \frac{\eps n}{2k} \right)}\\
&\leq \epsilon \delta/2k. 
\end{align*}

Therefore, on expectation, there are at most $n \cdot \epsilon \delta /4k$ nodes such that the number of neighbours in $S$ is more than $\epsilon n/2k$. Using Markov's inequality, with probability $1-\delta/2$, we have that at most $\epsilon n/2k$ nodes have  number of neighbours in $S$ greater than $\eps n/2k$. Applying union bound, we have with probability $1-\delta$ that $R$ is a $\eps/2k$-IS representative subset.
\end{proof}

\begin{lemma}
Suppose $S$ is an independent set in $G$ with cost $C(S) \geq \rho C(V)$ for some $\rho > 0$ and $\hat Z(\mathcal{U})$ denote the set found by Algorithm~\ref{alg:mis_tester} such that $\mathcal{U}$ is a $\epsilon/2k$-IS representative subset. Then, with probability $1-\delta$, we have $$C(\hat Z(\mathcal{U})) \geq \rho(1-3\gamma) C(V).$$
\label{lem:approx_appendix}
\end{lemma}
\begin{proof}
Consider $\hat Z_i(\mathcal{U})$ for some $i \in F_{(\gamma, \rho)}$ and let $F := \bigcup_{i \in F_{(\gamma, \rho)}} V_i$. In Algorithm~\ref{alg:mis_tester}, we obtained $\hat Z_i(\mathcal{U})$ by including nodes from $Z(\mathcal{U}) \cap V_i$ in the sorted order of degree such that the total degree of nodes in the induced graph $Z(\mathcal{U})$ is bounded by $\epsilon n^2/k$. First, when $\mathcal{U}$ is a $\epsilon/2k$-IS representative subset, we observe that 
$$\mathcal{U}  \subseteq S \cap F \subseteq F \setminus {\Gamma}(S \cap F) \subseteq Z(\mathcal{U}). $$ 

So, $S \cap V_i \subseteq Z(\mathcal{U}) \cap V_i$. From Lemma~\ref{lem:eps_representative}, we have, for every node in $Z(\mathcal{U}) \cap V_i$ except for $\epsilon/2k$ many, the maximum degree to $S \cap V_i$ is at most $\epsilon n /2k$, and the remaining nodes can have a maximum degree of $n$. Combining these statements, we have that the total degree incident on the nodes in $S\cap V_i$ from the nodes $Z(\mathcal{U}) \cap V_i$ is at most 
$$ \frac{\epsilon n}{2k} \cdot |Z(\mathcal{U}) \cap V_i| + n \cdot \frac{\epsilon n}{2k}   \leq \frac{\eps n^2}{k}.$$ 

As we include nodes in $\hat Z_i(\mathcal{U})$ until sum of degrees is $\epsilon n^2/k$, we have that the size of $\hat Z_i(\mathcal{U}_i)$ will only be more than the size of $S \cap V_i$ and satisfies $|\hat Z_i(\mathcal{U})| \geq |S \cap V_i|$. We know that every node in $V_i$ has cost in the range $[(1+\gamma)^{i-1}, (1+\gamma)^i)$, therefore, we have $$C(\hat Z_i(\mathcal{U})) \geq \frac{1}{(1+\gamma)} C(S \cap V_i)$$
\[  \sum_{i \in F_{(\gamma, \rho)}} C(\hat Z_i(\mathcal{U})) \geq \frac{1}{(1+\gamma)} \sum_{i \in F_{(\gamma, \rho)}} C(S \cap V_i) \]

From Claim~\ref{cl:approx_WIS_appendix}, we know 
\begin{align*}
    \sum_{i \in F_{(\gamma, \rho)}} C(\hat Z_i(\mathcal{U})) &\geq \frac{1}{(1+\gamma)} C(S')\\
    &\geq (1-\gamma)(1-2\gamma) \rho C(V) \\
C(\hat Z(\mathcal{U}))    &\geq \rho (1-3\gamma)  C(V).
\end{align*}
\end{proof}

\begin{lemma}
Let $G$ contain an independent set of cost $\rho C(V)$, then, Algorithm~\ref{alg:mis_tester} returns a set of nodes $\hat{Z}(\mathcal{U})$ such that $C(\hat{Z}(\mathcal{U})) \geq \rho(1- 3\gamma) C(V)$ and $|E[\hat{Z}(\mathcal{U})]| \leq \epsilon n^2$ with probability $1-\delta$ and runs in time $O\left( n^2 \exp{ \left(O\left(    \frac{k^2}{\eps} \log \frac{1}{\gamma \epsilon} \log \frac{k}{\epsilon \delta} \right)\right) } \right)$.
\label{lem:weighted_near_MIS_appendix}
\end{lemma}
\begin{proof}
As our Algorithm~\ref{alg:mis_tester} selects nodes from each partition $V_i$ such that the total degree of nodes in $\hat Z_i(\mathcal{U})$ in the graph induced by $E[{Z}(\mathcal{U})]$ is at most $\epsilon n^2/k$. Therefore, total degree of nodes in $\hat{Z}(\mathcal{U}) = \bigcup_{i \in F_{(\gamma, \rho)}} \hat Z_i(\mathcal{U})$ is at most $k \cdot \epsilon n^2/k$. Hence, $|E[\hat{Z}(\mathcal{U})]| \leq \epsilon n^2$. From Lemma~\ref{lem:approx_appendix}, we have $C(\hat Z(\mathcal{U})) \geq \rho (1-3\gamma) C(V)$.\\ \\
In Algorithm~\ref{alg:mis_tester}, we iterate over all subsets of the partitions $\{ V_1, V_2, \cdots V_k \}$. Consider a subset $\{ V_1, V_2, \cdots V_\tau \}$ and in each partition, we iterate over all subsets $U_i$ of size $\gamma \rho t/2$. Therefore, total number of subsets $\mathcal{U}$ formed from the union of subsets in each partition $\cup_{i=1}^\tau U_i$ is given by $\binom{t}{\gamma \rho t/2}^\tau$. Using $t = O(\frac{k \log k/\epsilon \delta}{\rho \gamma \epsilon})$ and $\rho \geq \sqrt{\epsilon}$, we have that the total number of iterations is at most

\begin{align*}
    2^k \cdot \binom{t}{\gamma \rho t/2}^k &\leq  2^k \cdot \left( \frac{2te}{\gamma \rho t}\right)^{\gamma \rho t k/2 }\\
    &\leq  2^k \cdot \left( \frac{6}{\gamma \rho} \right)^{\gamma \rho t k/2} \leq \exp{ \left(O\left(    \frac{k^2}{\eps} \log \frac{1}{\gamma \epsilon} \log \frac{k}{\epsilon \delta} \right)\right) }.
\end{align*}

In each iteration, we can find $Z(\mathcal{U})$ in $O(|\mathcal{U}| n)$ time. After that, we calculate the degree of nodes in $Z(\mathcal{U}) \cap V_i$ in the induced sub-graph $E[Z(\mathcal{U})]$ which requires $O({|Z(\mathcal{U})|}^2) = O(n^2)$ running time. Hence, the claim.
\end{proof}

\begin{lemma}
Suppose $S^*$ denotes MIS in $G(V, E)$. Algorithm~\ref{alg:near_mis} returns a set of nodes $S$ such that $C(S) \geq C(S^*)$, $|S| \geq \sqrt{\epsilon} n$ and $|E[S]| \leq \epsilon n^2$ with probability $1-\delta$ and runs in time $$O\left( \frac{n^2 W}{\eps} \log \frac{1}{\eps} \exp{ \left(O\left(    \frac{W^2 \log^2 W}{\eps^3} \log \frac{W}{\epsilon} \log \frac{W \log W \log 1/\eps}{{\epsilon \delta}} \right)\right) } \right).$$

\label{lem:approxMIS_appendix}
\end{lemma}
\begin{proof}
Let $T$ denote the set of $\sqrt{\eps}n$ nodes from $V$ with highest cost. It is easy to observe that $C(T) \geq \sqrt{\eps} \ C(V)$. If $C(S^*) < \ C(T)$, then, Algorithm~\ref{alg:near_mis} outputs the set $T$. Therefore, $$ C(T) > C(S^*) \text{ and }|E[T]| \leq (\sqrt{\eps} n)^2 = \eps n^2.$$
Otherwise, in Algorithm~\ref{alg:near_mis}, we search for MIS with cost $\rho C(V)$  using decreasing powers of $(1+\gamma)$ with the help of the parameter $\rho$ when $\rho \geq \sqrt{\eps}$. If $C(S^*) \geq C(T)$, then, $|S^*| \geq |T| = \sqrt{\eps} n$ and for some $1 \leq j \leq \frac{1}{2\gamma} \log \frac{1}{\eps}$ and $\rho = \frac{1}{(1+\gamma)^{j}}$ (i.e., $\sqrt{\eps} \leq \rho \leq 1$) we have $$\rho C(V) \leq  C(S^*) \leq \rho(1+\gamma)C(V).$$ 
For this value of $\rho$, Algorithm~\ref{alg:mis_tester} returns a set of nodes $S$ such that $|E[S]| \leq \frac{\eps}{8 W} n^2$. We observe that   
\begin{align*}
  C(S) &\geq \frac{1}{(1+\gamma)^j} (1-3\gamma) C(V) \\
       &\geq \frac{1-3\gamma}{1+\gamma} \frac{C(V)}{(1+\gamma)^{j-1}} \geq (1-4\gamma) C(S^*).
\end{align*}
In our call to the Algorithm~\ref{alg:mis_tester} from Algorithm \NearMIS, we set $\gamma = \frac{\eps}{8 W}$.
\begin{align*}
C(S_{\eps/2}) &\geq  \frac{\eps n}{2} \quad \text{(since, cost of a node is at least $1$)}\\
\Rightarrow C(S \cup S_{\eps/2}) &\geq C(S^*) + \frac{\eps n}{2} - \frac{\eps}{2 W} C(S^*)\\
&\geq C(S^*) + \frac{\eps n}{2} - \frac{\eps }{2 W} n \cdot W \geq C(S^*).
\end{align*}
As every node in $S_{\eps/2}$ has degree at most $n$, we have 
\[ |E[S \cup S_{\eps/2}]| \leq \frac{\eps n^2 }{8 W} + \frac{\eps n^2}{2} \leq \eps n^2.\]
As $C(S \cup S_{\eps/2}) \geq C(T)$ where $T$ contains the $\sqrt{\eps} n$ highest cost nodes, we have $|S \cup S_{\eps/2}| \geq \sqrt{\eps} n$. When $C(S^*) \geq C(T)$, we search for the correct value of $\rho$ and for each guess, we call the routine Algorithm~\ref{alg:mis_tester}. In total, the number of calls that are made to Algorithm~\ref{alg:mis_tester} is at most $\frac{1}{2\gamma} \log \frac{1}{\eps}$. However, in each call to Algorithm~\ref{alg:mis_tester}, we fail to output with probability $\delta'$. As we set the failure probability to $\delta' = 2\gamma\delta/\log(1/\epsilon)$, overall the iterations, using union bound, the failure probability is at most $\delta' \cdot \frac{1}{2\gamma} \log \frac{1}{\eps}  = \delta$. 

From Lemma~\ref{lem:weighted_near_MIS_appendix},
Algorithm~\ref{alg:mis_tester} runs in time $O\left( n^2 \exp{ \left(O\left(    \frac{k^2}{\eps} \log \frac{1}{\gamma \epsilon} \log \frac{k}{\epsilon \delta'} \right)\right) } \right)$. Substituting $k = \gamma^{-1} \log W,\delta'$, $\gamma = \frac{\eps}{8 W}$ and for a total of $\frac{1}{2\gamma} \log \frac{1}{\eps}$ calls to Algorithm~\ref{alg:mis_tester}, the running time of Algorithm~\ref{alg:near_mis} is 
\[ 
O\left( \frac{n^2 W}{\eps} \log \frac{1}{\eps} \exp{ \left(O\left(    \frac{W^2 \log^2 W}{\eps^3} \log \frac{W}{\epsilon} \log \frac{W \log W \log 1/\eps}{{\epsilon \delta}} \right)\right) } \right).
\]
\end{proof}

From Lemma~\ref{lem:approxMIS_appendix}, we have that in each iteration, Algorithm~\ref{alg:near_mis} returns a set of nodes $S$ that have a cost $C(S) \geq C(S^*)$ where $S^*$ is the maximum independent set in $G$. In a previous work~\citep{neurips18}, it was shown that by using maximum independent set in each iteration, we obtain a $(2 + \exp{(-\Omega(m))}$-optimal separating set system. Following the exact same analysis, gives us an approximation factor close to $2$. We refer the reader to the analysis in Appendix~\ref{app:2approx}, and give the main statement of the Lemma below.

\begin{lemma}
For any $m \geq \eta \log 1/\epsilon$ for some constant $\eta > 2$, with probability $\ge 1-\delta$, Algorithm~\ref{alg:ss_matrix} returns $L_\epsilon$ with $C(L_\epsilon) \leq (2 + \exp{(-\Omega(m))})  \cdot C(L^*)$, where $L^*$ is the min-cost separating matrix for $G$
\label{lem:2approx}
\end{lemma}

\noindent \textbf{Scaling Parameters}. In Algorithm~\ref{alg:ss_matrix}, we pass a scaled value of $\eps$ by setting it to $\eps^2$ when we call Algorithm~\ref{alg:near_mis}, as this ensures that total number of edges returned over $\frac{1}{\eps}$ calls is at most $\eps n^2$. We also set the failure probability for each call as $\eps \delta$, to ensure that over $\frac{1}{\eps}$ calls, total failure probability using union bound is at most $\delta$.

\begin{theorem}(Theorem~\ref{thm:main_ss} restated)
For any $m \geq \eta \log 1/\epsilon$ for some constant $\eta$, with probability $\ge 1-\delta$, Algorithm~\ref{alg:ss_matrix} returns $L_\epsilon$ with 
$C(L_\epsilon) \leq (2  + \exp{(-\Omega(m))})  \cdot C(L^*)$, where $L^*$ is the min-cost separating matrix for $G$. Moreover $L_\epsilon$ $\epsilon$-separates $G$. Algorithm~\ref{alg:ss_matrix} has a running time $O(n^2 f(W, \epsilon, \delta))$ where 
$f(W, \epsilon, \delta) = O\left( \frac{n^2 W}{\eps^2} \log \frac{1}{\eps} \exp{ \left(O\left(    \frac{W^2 \log^2 W}{\eps^6} \log \frac{W}{\epsilon} \log \frac{W \log W \log 1/\eps}{{\epsilon \delta}} \right)\right) } \right).$

\label{thm:main_ss_appendix}
\end{theorem}
\begin{proof}
Using all the above scaled parameters, from Lemma~\ref{lem:approxMIS_appendix}, the running time of Algorithm~\ref{alg:ss_matrix} that internally calls Algorithm~\ref{alg:near_mis} for $\frac{1}{\eps}$ number of times, is given by 
\[ 
O\left( \frac{n^2 W}{\eps^2} \log \frac{1}{\eps} \exp{ \left(O\left(    \frac{W^2 \log^2 W}{\eps^6} \log \frac{W}{\epsilon} \log \frac{W \log W \log 1/\eps}{{\epsilon \delta}} \right)\right) } \right).
\]
From Lemma~\ref{lem:2approx}, we have the approximation guarantee.
\end{proof}

\noindent \textbf{Remark}. Observe that our running time is exponential in $1/\eps$ and therefore setting $\eps < 1/n^2$ to get a separating system with all edges separated requires exponential running time. {As we have argued that  finding such a set system with near optimal cost is hard conditioned on the hardness of approximate coloring (Theorem~\ref{thm:ss_hardness}), it is thus also conditionally hard to improve our runtime to be polynomial in $1/\epsilon$.} It is an interesting open question to study the parameterized hardness beyond polynomial factors with respect to $\epsilon$.



By Theorem \ref{thm:main_ss} with $m = O(\log(1/\epsilon))$ interventions we can $\epsilon$-approximately learn any causal graph $G$. For learning the entire graph $G$, $m \geq \log \chi$ interventions are necessary, where $\chi$ is the chromatic number of $G$, since the rows of $L \in \{0,1\}^{n\times m}$ must be a valid coloring of $G$~\citep{neurips18}.
\section{Missing Details From Section~\ref{sec:ancG}}\label{app:ancG}

In this section, we say a pair of nodes $(v_i, v_j)$ share an ancestral relation, if $v_i$ has a directed path to $v_j$ ($v_i$ is an ancestor of $v_j$) or $v_j$ has a directed path to $v_i$ ($v_j$ is an ancestor of $v_i$).

\begin{lemma}
Suppose $\mathcal{S} = \{ S_1, S_2, \cdots, S_m \}$ is a collection of subsets of $V$. If $\Anc(G) \cap H$ is recovered from $H$ using conditional independence tests by intervening on the sets $S_i \in \mathcal{S}$. Then, under the assumptions of section~\ref{sec:prelim}, $\mathcal{S}$ is a strongly separating set system on $H$.
\end{lemma}
\begin{proof}
First, we argue that to recover $\Anc(G) \cap H$ it is sufficient that $\mathcal{S}$ is a strongly separating set system on $H$. Suppose $(v_i, v_j) \in H$ and $v_i, v_j$ share an ancestral relation i.e., either $v_i$ is an ancestor of $v_j$ or $v_j$ is an ancestor of $v_i$. Therefore, $v_i \notindep v_j$ and $(v_i, v_j) \in \Anc(G)$. From Lemma 1 in~\cite{neurips17}, we know that, we can recover the ancestral relation between $v_i$ and $v_j$ using conditional independence tests (or CI-tests) on interventional distributions that strongly separate the two variables $v_i$ and $v_j$. As $\mathcal{S}$ is a strongly separating set system for $H$, we can recover all ancestral relations in $\Anc(G) \cap H$.

Now, we show that a strongly separating set system on $H$ is necessary. Here, we give a proof similar to Lemma A.1 from~\cite{AKMM2020}. Suppose $\mathcal{S}$ is not a strongly separating set system for $H$. If there exists a pair of nodes containing an ancestral relation, say $(v_i, v_j) \in H \cap \Anc(G)$ such that every set $S_k \in \mathcal{S}$ contains none of them, then, we cannot recover the ancestral relation between these two nodes as we are not intervening on either $v_i$ or $v_j$ and the results of an independence test $v_i \indep v_j$ might result in a wrong inference, possibly due to the presence of a latent variable $l_{ij}$ between them. Consider the case when only one of them is present in every set of $\mathcal{S}$. Let $\mathcal{S}$ be such that $\forall S_k \in \mathcal{S} :  S_k \cap \{ v_i, v_j \}  = \{ v_i \} \Rightarrow v_i \in S_k, v_j \not \in S_k $. We choose our graph $G$ to have two components $\{v_i, v_j\}$ and $V \setminus \{v_i, v_j\}$; and include the edge $v_j \rightarrow v_i$ in it. Observe that $v_i \notindep v_j$.  Our algorithm will conclude from the CI-test $v_i \indep v_j \mid \doo(S_k) ?$ that $v_i$ and $v_j$ are independent. However, it is possible that $v_i \notindep v_j$ because of a latent $l_{ij}$ between $v_i$ and $v_j$, but when we do CI-test, we get $v_i \indep v_j \mid \doo(S_k)$ as intervening on $v_i$ disconnects the $l_{ij} \rightarrow v_i$ edge. Therefore, our algorithm cannot distinguish the two cases $v_j \rightarrow v_i$ and $v_i \leftarrow l_{ij} \rightarrow v_j$ without intervening on $v_j$. For every $\mathcal{S}$ that is not a strongly separating set system on $H$, we can provide a graph $G$ such that by intervening on sets in $\mathcal{S}$, we cannot recover $\Anc(G) \cap H$ from $H$ correctly.
\end{proof}

\begin{lemma}(Lemma~\ref{clm:ssssSuffices} restated)
Under the assumptions of Section \ref{sec:prelim}, if $\mathcal{S} = \{ S_1, S_2, \cdots S_m \}$ is an $\epsilon$-strongly separating set system for $H$, $\mathcal{S}$ suffices to $\epsilon$-approximately learn $\Anc(G) \cap H$.
\end{lemma}
\begin{proof}
Given $\mathcal{S}$ denotes an $\epsilon$-strongly separating set system for $H$. So, there are at most $\epsilon n^2$ edges $(u, v) \in H$ such that for all $i \in [m]$, either $\{u, v\} \cap S_i = \phi$ or $\{u, v\} \cap S_i = \{u, v\}$ or $\{ u, v \} \cap S_i = \{ u \}$ (without loss of generality). For every such edge, any intervention on a set in $\mathcal{S}$, say $S_i$ cannot recover the direction from a conditional independence test(CI-test) $u \indep v \mid \doo(S_i)$ because for both the cases $u \leftarrow v$ or $u \leftarrow l_{uv} \rightarrow v$, where $l_{uv}$ is a latent, the CI-test returns that they are independent. Therefore, we cannot recover the ancestral relation (if one exists) between $u, v$ that are not strongly separated in $H$. For edges $(u,v) \in H$ that are strongly separated using $S_i$ and $S_j$, we can recover the ancestral relation using CI-tests $u \indep v \mid \doo(S_i)$ and $u \indep v \mid \doo(S_j)$. From Def.~\ref{def:eps_learn}, we have that $\mathcal{S}$ $\epsilon$-approximately learns $G$. 
\end{proof}

Algorithm \textsc{SSMatrix} from \cite{AKMM2020} gives a $2$-approximation guarantee for the output strongly separating matrix. However, we cannot directly extend the arguments as the guarantee holds when the input graph is complete. We get around this limitation, and show that Algorithm~\ref{alg:sss_matrix} achieves a close to $4$-approximation, by relating the cost of $\eps$-strongly separating matrix returned by \textsc{SSMatrix} on the \textit{supernode} set $V_S$, to the cost of $2$-approximate $\eps$-separating matrix that we find using Algorithm~\ref{alg:ss_matrix}.   

Let $\ALG_{S}$ denote the cost of the objective $\sum_{j=1}^n C(v_j)\, \| L^S_\eps(j) \|_1$ obtained by Algorithm~\ref{alg:ss_matrix} where $L^S_\eps$ is an $\eps$-separating matrix; $\ALG_{SS}$ denote the cost of the objective $\sum_{j=1}^n C(v_j)\, \| L^{SS}_\eps(j) \|_1$ obtained by Algorithm~\ref{alg:sss_matrix} where $L^{SS}_\eps$ is an $\eps$-strongly separating matrix. For the sake of analysis, during assignment of vectors to nodes in $L^S_{\eps}$, we assume that Algorithm~\ref{alg:ss_matrix} only allows vectors of weight at least $1$. As \textsc{SSMatrix} algorithm from~\cite{AKMM2020} assigns vectors with weight atleast $1$ (otherwise it will not be a valid strongly separating matrix), this assumption for $\ALG_S$ helps us in showing a relation between the costs of $L^{SS}_\eps$ and $L^S_\eps$. As that is not sufficient to obtain the claimed guarantee, instead of assigning ${m \choose 1}$ vectors of weight $1$, we constraint it to a fixed number $r \leq {m \choose 1}$. In~\cite{AKMM2020}, Algorithm \textsc{SSMatrix} assigns vectors to $L^{SS}_\eps$ by guessing the exact number of weight $1$ vectors in $\OPT_{SS}$, the parameter $r$ corresponds to this guess. 

Let ${\OPT}_{SS}$ and $\OPT_{S}$ denote optimum objective values associated with strongly separating and separating matrices for a graph $H$. Let $\ALG_{SS}(r)$ denote the cost $C(L^{SS}_{\eps})$ assuming first $r$ columns are used for exactly $r$ weight $1$ vectors during the assignment in $L^{SS}_{\eps}$, and the remaining $m-r$ columns are used for all the remaining vector assignments. Similarly, $\ALG_S(r)$, $\OPT_S(r)$ \text{ and }$\OPT_{SS}(r)$ are defined. 

\begin{lemma}
$\OPT_{S}(r) \leq \OPT_{SS}(r)$ for any $r \geq 0$.
\label{lem:ss_lb_appendix}
\end{lemma}
\begin{proof}
Observe that any strongly separating matrix for $H$ is also a separating matrix for $H$. Now, consider a strongly separating matrix that achieves cost $\OPT_{SS}(r)$ using $r$ weight $1$ vectors, then, we have $$\OPT_{S}(r) \leq \OPT_{SS}(r).$$
\end{proof}
\begin{lemma} $C(L^{SS}_{\epsilon}) \leq (4+\gamma+ \exp{(-\Omega(m))}) \cdot {\OPT}_{SS}.$ 
\label{lem:approx_SS_appendix}
\end{lemma}
\begin{proof}
First, we note that in \textit{any} strongly separating matrix, for the \emph{non-dominating} property to hold, the support of weight $1$ vectors and the support of vectors of weight $ > 1$ are column disjoint. 
 Suppose $a_1^*$ denote the number of columns of $m$ that are used by $\OPT_{SS}$ for weight $1$ vectors i.e, $\OPT_{SS}(a^*_1) = \OPT_{SS}$. 
 
 Following the exact proof of Lemma A.5 in \cite{AKMM2020} gives us the following guarantee about Algorithm~\ref{alg:sss_matrix} $$ C(L^{SS}_{\eps}) \leq 2\ALG_S(a_1^*).$$

From Theorem~\ref{thm:main_ss} and Lemma~\ref{lem:2approx} in Appendix~\ref{app:2approx} (or the analysis from~\citet{neurips18}) :
$$\ALG_S(a_1^*) \leq (2+\exp{(-\Omega(m))} ) \OPT_S(a_1^*).$$
From Lemma~\ref{lem:ss_lb_appendix}, we know $\OPT_S(a^*_1) \leq \OPT_{SS}(a^*_1)$. Therefore, we have
 \begin{align*}
     \ALG_S(a_1^*) &\leq (2+\exp{(-\Omega(m))})\OPT_{SS}(a_1^*)\\
   &= (2+ \exp{(-\Omega(m))} ) \OPT_{SS}.
 \end{align*}
Hence, the lemma.
\end{proof}

\begin{theorem}(Theorem~\ref{thm:main_sss} restated)
Let $m \geq \eta \log 1/\epsilon$ for some constant $\eta$ and $L^{SS}_{\epsilon}$ be matrix returned by Algorithm~\ref{alg:sss_matrix}. Then with probability $\ge 1-\delta$, $L^{SS}_{\epsilon}$ is an $\epsilon$-strongly separating matrix for $H$ and $C(L_\epsilon^{SS}) \le (4+\exp{(-\Omega(m))})   \cdot C(L^*)$
where $L^*$ is the min-cost strongly separating matrix for $H$. Algorithm~\ref{alg:sss_matrix} runs in time 
$O(n^2 f(W, \epsilon, \delta))$ where $$f(W, \epsilon, \delta) = O\left( \frac{n^2 W}{\eps^2} \log \frac{1}{\eps} \exp{ \left(O\left(    \frac{W^2 \log^2 W}{\eps^6} \log \frac{W}{\epsilon} \log \frac{W \log W \log 1/\eps}{{\epsilon \delta}} \right)\right) } \right).$$
\label{thm:main_sss_appendix}
\end{theorem}
\begin{proof}
From Lemma~\ref{lem:approx_SS_appendix}, we have $C(L_\epsilon^{SS}) \leq (4+ \exp{(-\Omega(m))}) C(L^*)$. The sets of nodes $S_1, S_2, \cdots S_{1/\epsilon}$ returned by Algorithm~\ref{alg:sss_matrix} are such that every set $S_i$ contains at most $\epsilon^2 n^2$ edges with probability $1-\delta$. Therefore, in total at most $\frac{1}{\epsilon} \epsilon^2 n^2 \leq \epsilon n^2$ edges do not satisfy strongly separating property. 
As \textsc{SSMatrix} has a running time of $O(|V_S|) = O(\frac{1}{\epsilon})$, and using the running time of Algorithm~\ref{alg:ss_matrix} from Theorem~\ref{thm:main_ss_appendix}, our claim follows. 
\end{proof}

\section{Hyperfinite Graphs : Better Guarantees}\label{app:hyperfinite}
In this section, we consider the special case of degree bounded hyperfinite graphs. 

\begin{definition}
A Graph $G(V, E)$ is $(\eps, k)$-hyperfinite if there exists $E' \subseteq E$ and $|E' \setminus E| \leq \eps n$ such that every connected component in the induced subgraph of $E'$ is of size at most $k$. A Graph $G$ is said to be $\tau$-hyperfinite, if there exists a function $\tau : \mathbb{R}^+ \rightarrow \mathbb{R}^+$ such that for every $\eps > 0$, $G$ is $(\eps, \tau(\eps))$-hyperfinite.
\end{definition}

If a $\tau$-hyperfinite graph $G$ has maximum degree $\Delta$,  we give algorithms for (strongly) separating set systems on $G$ that obtain the same approximation guarantees, but the number of edges that are not (strongly) separated at most $\epsilon \cdot n \cdot \Delta$. In order to obtain that, we extend the additive approximation algorithm of ~\cite{hassidim2009local} for finding the maximum independent set to the weighted graphs i.e., when the nodes have costs and return a $\textsc{Near-MIS}$ instead of MIS. First, we define a very important partitioning of the nodes $V$ possible in $\tau$-hyperfinite graphs and give the lemma that describes the guarantees associated with finding the partitions.

\begin{definition}[\cite{hassidim2009local}$ (\eps, k)$ partitioning oracle $\mathcal O$] For a given graph $G(V, E)$ and query $q$ about $v \in V$, $\mathcal O$ returns the partition $P[v] \subseteq V $ containing $v$ that satisfies :  
\begin{enumerate}
    \item for every node $v \in V$, $|P[v]| \leq k$ and $P[v]$ is connected
    \item $|\{ (u, w) \in E \mid P[u] \neq P[w] \}| \leq \eps \cdot n$ with probability $9/10$.
\end{enumerate}
\end{definition}

\begin{lemma}[\cite{hassidim2009local}]\label{lem:partition}
If $G$ is $(\eps, \tau(\eps))$-hyperfinite graph with maximum degree $\Delta$, then, there is a $(\eps \cdot \Delta, \tau(\eps^3/54000))$ partition oracle that answers a given query $q$  with probability $1-\delta$, using a running time $O({2^{\Delta^{O(\tau(\eps^3))}}}/{\delta} \log {1}/{\delta})$.
\end{lemma}

Using Lemma~\ref{lem:partition}, we query every node to obtain the partitioning of $V$ and formalize this in the following corollary.
\begin{corollary}\label{cor:partition}
If $G$ is $(\eps, \tau(\eps))$-hyperfinite graph with maximum degree $\Delta$, then, we can obtain a partitioning of the graph $G$, given by $V_1, V_2, \cdots $ such that with probability $1-\delta$ and a running time of $O(\frac{n}{\delta} \cdot {2^{\Delta^{O(\tau(\eps^3/\Delta^3))}}} \log {1}/{\delta})$, we have :
\begin{enumerate}
    \item For every $i, |V_i| \leq \tau(\eps^3/\Delta^3 54000)$ and $V_i$ is connected
    \item $|\{ (u, w) \mid (u, w) \in E, u \in V_i, w \in V_j \text{ and } i \neq j \}| \leq \eps \cdot n$
\end{enumerate}
\end{corollary}

Given a $\tau$-hyperfinite graph with maximum degree $\Delta$, we describe an algorithm that returns a set of nodes that have at most $\eps n \Delta$ edges instead of $\eps n^2$ edges that we saw previously for general graphs $G$. To do so, we build upon the previous result from~\cite{hassidim2009local} that returns a set of nodes $R$ which is an additive $\eps n$ approximation of MIS $S^*$, i.e., $|R| \geq |S^*| - \eps n$. To obtain this, the authors first use the partitioning obtained using Lemma~\ref{lem:partition} and find MIS in each partition separately. They show that ignoring the nodes that are incident on edges across the partitions obtained using Lemma~\ref{lem:partition} will only result in a loss of $\eps n$ nodes. We observe that in Algorithm~\ref{alg:nearMIS_hyperfinite}, by removing the $\eps \cdot n$ nodes (denoted by $\Hat{V}$) that are incident on the edges across partitions, and adding back $\eps n$ nodes with highest cost, we will obtain a set of nodes with cost at least that of MIS while only adding $\eps n \Delta$ edges amongst the combined set of nodes. 
\begin{algorithm}[h]
\caption{$\textsc{Near-MIS}$ in  $\tau$-Hyperfinite Graph $G$}
\label{alg:nearMIS_hyperfinite}
\begin{algorithmic}[1]
\State \textbf{Input} : Graph $G = (V, E)$, cost function $C:V \rightarrow \R^+$, $m$, $\Delta$, function $\tau(\cdot)$, error $\epsilon$, failure probability $\delta$.
\State \textbf{Output} : $T$ that is a $\textsc{Near-MIS}$ with at most $\eps \cdot n \cdot \Delta$ edges.
\State Let the set of partitions is $\{V_1, V_2 \cdots \}$ of $G(V, E)$ returned using Corollary~\ref{cor:partition} with parameters $\tau(\cdot)$, error $\eps/2$ and failure probability $\delta$.
\For{each partition $V_i$}
\State Calculate the maximum \emph{cost} independent set $T_i$ in $V_i$.
\EndFor
\State $\Hat{E} \leftarrow \{ (u, v) \mid (u, v) \in E \text{ and there exists } i, j \text{ where } i\neq j, u \in V_i, v\in V_j  \}$.
\State $\Hat{V} \leftarrow \{ u \mid \exists v \text{ such that } (u, v) \in \Hat{E} \}$.
\State $T \leftarrow \left( \bigcup_{i=1} T_i \right) \setminus \Hat{V}$. 
\State Let $H$ denote $\eps \cdot n$ nodes of highest cost in $V \setminus T$. 
\State \Return $T \cup H$.
\end{algorithmic}
\end{algorithm}

\begin{lemma}\label{lem:nearMIS_hyperfinite}
In Algorithm~\ref{alg:nearMIS_hyperfinite}, we have $|\Hat{V}| \leq \eps \cdot n$ and $C(T \cup H) \geq C(S^*).$
\end{lemma}
\begin{proof}

From Corollary~\ref{cor:partition}, we have $|\Hat{E}| \leq \eps \cdot n/2$. From the definition of $\Hat{V}$, we have \[|\Hat{V}| \leq 2 |\Hat{E}| \leq \eps \cdot n.\]

Suppose $S^*$ is the maximum cost independent set in $G$. Now, consider all nodes in $\Hat{V}$. Similar to the above case, it is possible that $(u, w) \in \Hat{E}$ and $u \in T_i$, $w \in T_j$ for some $i \neq j$.
Consider a node $u \in S$ that is isolated in $E' \subseteq E$, and included in some partition $V_i$. 

As $T_i$ is maximum cost independent set in $V_i$, we have $C(T_i) \geq C(S^* \cap V_i)$ where $S^* \cap V_i$ is an independent set induced by MIS $S^*$ in the partition $V_i$. Combining it for all partitions, we have $C(\bigcup_i T_i) \geq C(S^*)$. As nodes in $\Hat{V}$, it is possible that including those that share an edge in $\bigcup_i T_i$ will result in the set of nodes not forming an independent set. However, the set $\bigcup_i T_i \setminus \Hat{V}$ formed by removing all the nodes that are incident with edges across the partitions, is an independent set. Since $S^*$ is MIS, we have
$$C(\bigcup_i T_i \setminus \Hat{V}) \leq C(S^*).
$$
As $|\Hat{V}|$ is at most $\eps \cdot n$, replacing them with $H$ consisting of $\eps \cdot n$ highest cost nodes from $T$ will only increase the cost. Therefore, we have
$$\Rightarrow           C(T \cup H) = C\left( \left(\bigcup_i T_i \setminus \Hat{V} \right) \cup H \right)\geq C(\bigcup_i T_i) \geq  C(S^*).$$

\end{proof}

\begin{theorem}\label{thm:nearMIS_hyperfinite}
Algorithm~\ref{alg:nearMIS_hyperfinite} returns a set $T \subseteq V$ of nodes such that $C(T) \geq C(S^*)$ where $S^*$ is the maximum cost independent set; $|E[T] | \leq \eps \cdot n \cdot \Delta$ and uses a running time $O(\frac{n}{\delta} \cdot {2^{\Delta^{O(\tau(\eps^3/\Delta^3))}}} \log {1}/{\delta} + n \Delta)$ with probability $1-\delta$
\end{theorem}
\begin{proof}
From Lemma~\ref{lem:nearMIS_hyperfinite}, we have $C(T) \geq C(S^*)$ and the nodes in $H$ include $\eps \cdot n$ nodes that are added (line 10 in Algorithm~\ref{alg:near_mis})  have at most $\eps \cdot n \cdot \Delta$ edges among themselves. Therefore,  $|E[T]|\leq \eps \cdot n \cdot \Delta$. Using Corollary~\ref{cor:partition}, we have that it takes $O(\frac{n}{\delta} \cdot {2^{\Delta^{O(\tau(\eps^3/\Delta^3))}}} \log {1}/{\delta})$ time to find the partitions. After finding the partitions, we find maximum cost independent set in each of the at most $n$ partitions each of size $O(\tau(\eps^3/\Delta^3))$, which takes a running time of $$O(\text{finding maximum cost independent set in each partition}) = O(n \cdot 2^{O(\tau(\eps^3/\Delta^3))}).$$
Combining the running times for both these steps, along with $O(n \Delta)$, the time to find $\Hat{E}$, we have the running time as claimed.
\end{proof}

We can use Algorithm~\ref{alg:nearMIS_hyperfinite} to obtain \NearMIS~in each iteration of Algorithm~\ref{alg:ss_matrix}; from Lemma~\ref{lem:2approx} and Theorem~\ref{thm:nearMIS_hyperfinite}, we have the following proposition about \emph{separating set system} for $G$.
\begin{proposition}\label{prop:2approx_hyperfinite}
Let $G(V, E)$ be a $\Delta$-degree bounded $\tau$-hyperfinite graph. For any $m \geq \eta \log 1/\epsilon$ for some constant $\eta$, with probability $\ge 1-\delta$, there is an algorithm that returns $L_\epsilon$ with 
$C(L_\epsilon) \leq (2 + \exp{(-\Omega(m))})  \cdot C(L^*)$, where $L^*$ is the min-cost separating matrix for $G$ and has a running time $O\left(\frac{n^3}{\delta} \cdot {2^{\Delta^{O(\tau(\eps^3/n^3 \Delta^3))}}} \log \frac{n}{\delta}\right)$. Moreover using $L_\epsilon$, the number of edges that are not separated in $G$ is at most $\eps \cdot n \cdot \Delta$.
\end{proposition}
\begin{proof}
In every iteration, we identify a set of nodes that has the cost at least the cost of MIS. Therefore, total number of iterations possible is at most $n$. Scaling the error parameter by setting $\epsilon' = \eps/n$ and $\delta' = \delta/n$ for each iteration, we have that Algorithm~\ref{alg:ss_matrix} returns $L_{\eps}$ such that the number of edges that are not separated is $n \cdot (\epsilon' \cdot n \cdot \Delta) = \epsilon \cdot n \cdot \Delta$. Using Lemma~\ref{lem:nearMIS_hyperfinite}, we have that the total running time of our algorithm is 
$$O(n \cdot \frac{n}{\delta'} \cdot {2^{\Delta^{O(\tau(\epsilon'^3/\Delta^3))}}} \log {1}/{\delta'}) = O\left(\frac{n^3}{\delta} \cdot {2^{\Delta^{O(\tau(\eps^3/n^3 \Delta^3))}}} \log \frac{n}{\delta}\right).$$
\end{proof}

We can obtain a similar result for \emph{strongly separating set system} for $G$ using Algorithm~\ref{alg:sss_matrix} and give the following proposition.
\begin{proposition}
Let $G(V, E)$ be a $\Delta$-degree bounded $\tau$-hyperfinite graph. For any $m \geq \eta \log 1/\epsilon$ for some constant $\eta$, with probability $\ge 1-\delta$, there is an algorithm that returns $L_\epsilon$ with 
$C(L_\epsilon) \leq (4 + \exp{(-\Omega(m))})  \cdot C(L^*)$, where $L^*$ is the min-cost \emph{strongly} separating matrix for $G$ and has a running time $O\left(\frac{n^3}{\delta} \cdot {2^{\Delta^{O(\tau(\eps^3/n^3 \Delta^3))}}} \log \frac{n}{\delta}\right)$. Moreover using $L_\epsilon$, the number of edges that are not \emph{strongly} separated in $G$ is at most $\eps \cdot n \cdot \Delta$.
\end{proposition}
\begin{proof}
In Algorithm~\ref{alg:sss_matrix}, we first find $L^S_\eps$, a separating matrix obtained using Proposition~\ref{prop:2approx_hyperfinite} that does not separate $\eps \cdot n \cdot \Delta$ edges of $G$. Next, we find \emph{super nodes} using the~\NearMIS's~returned and assign it vectors appropriately to form strongly separating matrix $L^{SS}_{\eps}$ on super nodes. Using Theorem~\ref{thm:main_sss_appendix}, we have the claimed approximation guarantee. The running time follows from Proposition~\ref{prop:2approx_hyperfinite}.
\end{proof}

\section{Additional Details for the analysis of $2$-approximation result for $\epsilon$-Separating Set Systems}\label{app:2approx}

In this section, we present already known results from ~\cite{neurips18} filling in the details in the analysis of our Algorithm~\ref{alg:ss_matrix} for the sake of completion. 

\noindent Let $\mathcal{I}$ denote the set of all independent sets in $G$. For some $\mathcal{A} \subseteq \mathcal{I}$, we have
\[
J(\mathcal{A}) = \sum_{v \in \bigcup_{S \in \mathcal{A}}} \ C(v)
\]
that is, it takes a set of independent sets and returns the sum of the cost of the vertices in their union. We observe that $J$ is submodular, monotone, and non-negative~\citep{neurips18}. 

Let $S_0$ denote the set of nodes that are assigned weight $0$ vector after the first iteration of Algorithm~\ref{alg:ss_matrix}. We set $V = V \setminus S_0$ for the remainder of this section and handle the cost contribution of nodes in $S_0$ separately in the analysis of approximation ratio.

\begin{lemma}\label{lem:submodular_appendix}
Given a submodular, monotone, and non-negative function $J$ over a ground
set $V$ and a cardinality constraint $k$. Let Algorithm~\ref{alg:ss_matrix} return $S_{\text{greedy}}$ a collection of at most $C k$ (for some constant $C > 0$) sets that are $(0, \eps)$-\NearMIS, then, $$J(S_{\text{greedy}}) \geq (1-e^{-C}) \max_{\mathcal{S} \subseteq \mathcal{I}, |\mathcal{S}| \leq k } J(\mathcal{S}).$$
\end{lemma}
\begin{proof} We have that in $i$th iteration of Algorithm~\ref{alg:ss_matrix}, we pick a set of nodes $S_i$ with cost at least the cost of MIS $T_i$ in $G$, i.e., $S_i$ is a $(0, \epsilon)$-\NearMIS~in $G$ and satisfies $C(S_i) \geq C(T_i)$. Let $\mathcal{S}^*$ be the collection of independent sets such that $J(\mathcal{S}^*) = \max_{\mathcal{S} \subseteq \mathcal{I}, |\mathcal{S}| \leq k } J(\mathcal{S})$. Let $\bigcup_{ j \leq i} S_j$ be denoted by $S_{1:i}$. We claim using induction that 
\[ 
    J(\mathcal{S}^*) - J(S_{1:i}) \leq \left( 1-\frac{1}{k} \right)^i J(\mathcal{S}^*).
\]
Consider $i$th iteration when Algorithm~\ref{alg:ss_matrix} picks $S_i$. Using submodularity, we have 
\[ 
     J(\mathcal{S}^*) - J(S_{1:i-1}) \leq \sum_{B \in \mathcal{S}^* \setminus S_{1:i-1}} J( S_{1:i-1} \cup B).
\]

Therefore, there exists one set $B \in \mathcal{S}^* \setminus S_{1:i-1}$, with cost at least $\frac{\sum_{B \in \mathcal{S}^* \setminus S_{1:i-1}} J( S_{1:i-1} \cup B)}{k}.$ As argued in Lemma~\ref{lem:approxMIS_appendix}, we are picking a set $S_i$ with cost $C(S_i) \geq C(T_i)$ where $T_i$ is MIS in the $i$th iteration, we have :

\begin{align*}
    C(S_i) \geq C(T_i) \geq \frac{\sum_{B \in \mathcal{S}^* \setminus S_{1:i-1}} J( S_{1:i-1} \cup B)}{k} &\geq  \frac{J(\mathcal{S}^*) - J(S_{1:i-1})}{k}\\
    J(S_i) = C(S_i) &\geq \frac{J(\mathcal{S}^*) - J(S_{1:i-1})}{k}.
\end{align*}

For $i = 1$, our claim follows from the above statement, i.e., $C(S_1) = J(S_1) \geq \frac{J(\mathcal{S}^*)}{k}$. Assuming that our claim holds until iteration $i-1$ for some $i \geq 2$, we have after the $i$th iteration : $J(\mathcal{S}^*) - J(S_{1:i}) = J(\mathcal{S}^*) - J(S_{1:i-1}) - J(S_i)$. This is true because $J(S_{1:i}) = J(S_{1:i-1}) + J(S_i)$ as $S_i$ is greedily chosen by picking a set containing nodes that are not previously selected. Therefore, 
\begin{align*}
    J(\mathcal{S}^*) - J(S_{1:i}) &= J(\mathcal{S}^*) - J(S_{1:i-1}) - J(S_i)\\
                                  &\leq J(\mathcal{S}^*) - J(S_{1:i-1}) - \frac{J(\mathcal{S}^*) - J(S_{1:i-1})}{k}\\
                                  &\leq \left(J(\mathcal{S}^*) - J(S_{1:i-1})\right)\left( 1 - \frac{1}{k} \right)\\
                                  &\leq \left( 1-\frac{1}{k} \right)^{i} J(\mathcal{S}^*).
\end{align*}
Setting $i = C \cdot k$, we have $$J(\mathcal{S}^*) - J(S_{\text{greedy}}) \leq \left( 1-\frac{1}{k} \right)^{C k} J(\mathcal{S}^*) \leq e^{-C} \cdot J(\mathcal{S}^*)$$
$$\Rightarrow J(S_{\text{greedy}}) \geq (1-e^{-C}) J(\mathcal{S}^*).$$
\end{proof}

Now, we define two types of submodular optimization problem, called the submodular chain problem and the supermodular chain problem that will be useful later. 

\begin{definition}
Given integers $k_1, k_2, \ldots, k_m$ and a submodular, monotone, and non-negative function $J$, over a ground set $V$, the \textit{submodular chain problem} is to find sets $A_1, A_2, \ldots, A_m \subseteq 2^{[V]}$ such that $\vert A_i \vert \leq k_i$ that maximizes
\[
\sum_{i = 1}^m J(A_1 \cup A_2, \cup \cdots \cup A_i).
\]
\end{definition}

\begin{lemma}\label{lem:submodular_chain}
Let $A_1^*, A_2^*, \ldots, A_m^*$ be the optimal solution to the submodular chain problem. Suppose that for all $1 \leq p \leq m / 2 - 1$ we have that $\sum_{i = 1}^{2p} k_i \geq \tau \sum_{i=1}^p k_i$. Also assume that $J(A_1 \cup A_2 \cup \cdots \cup A_m) = J(V)$. Then the greedy algorithm~\ref{alg:ss_matrix} for the submodular chain problem returns set $A_1, A_2, \ldots, A_m$ such that
\[
\sum_{i=1}^{m} J(A_1 \cup A_2 \cdots A_{i}) \geq J(V) + 2(1 - e^{-\tau})\sum_{i=1}^{m/2-1}J(A^*_1 \cup A^*_2 \cup \cdots A^*_i).
\]
\end{lemma}
\begin{proof}
Given $\sum_{i = 1}^{2p} k_i \geq \tau \sum_{i=1}^p k_i$. From Lemma~\ref{lem:submodular_appendix}, we have
\[
J(A_1 \cup A_2 \cup \cdots A_{2p}) \geq (1 - e^{-\tau}) J(A^*_1 \cup A^*_2 \cup \cdots \cup A^*_p).
\]

\[
\sum_{i=1}^{m/2 - 1}J(A_1 \cup A_2 \cdots A_{2i}) \geq (1 - e^{-\tau})\sum_{i=1}^{m/2 - 1}J(A^*_1 \cup A^*_2 \cup \cdots \cup A^*_i).
\]
Now, we use the monotonicity property of the submodular function $J$ to get
\begin{align*}
    \sum_{i=1}^m J(A_1 \cup A_2 \cdots A_{i})  &= J(A_1 \cup A_2 \cdots \cup A_m) +\sum_{i=1}^{m/2 -1} J(A_1 \cup A_2 \cup \cdots \cup A_{2i}) + J(A_1 \cup A_2 \cup \cdots \cup A_{2i+1}) \\
    &\geq J(V) + 2\sum_{i=1}^{m/2 -1} J(A_1 \cup A_2 \cup \cdots \cup A_{2i}).
\end{align*}
Hence, the lemma.
\end{proof}

\begin{definition}
Given integers $k_1, k_2, \ldots, k_m$ and a submodular, monotone, and non-negative function $F$, over a ground set $V$, the \textit{supermodular chain problem} is to find sets $A_1, A_2, \ldots, A_m \subseteq 2^{[V]}$ such that $\vert A_i \vert \leq k_i$ that minimizes
\[
\sum_{i=1}^m J(V) - J(A_1 \cup A_2 \cup \cdots \cup A_i).
\]
\end{definition}
For the greedy algorithm~\ref{alg:ss_matrix}, we give the following claim for the supermodular chain problem.

\begin{lemma}\label{lem:supermodular-chain}
Let $A_1^*, A_2^*, \ldots, A_m^*$ be the optimal solution to the supermodular chain problem. Suppose that for all $1 \leq p \leq m / 2 - 1$ we have that $\sum_{i = 1}^{2p} k_i \geq \tau \sum_{i=1}^p k_i$. Also assume that $J(A_1 \cup A_2 \cup \cdots \cup A_m) = J(V)$. Then the greedy algorithm~\ref{alg:ss_matrix} for the supermodular chain problem returns set $A_1, A_2, \ldots, A_m$ such that
\[
\sum_{i=1}^m J(V) - J(A_1 \cup A_2 \cup \cdots \cup A_i) \leq e^{-\tau} m \cdot J(V) + 2\sum_{i=1}^{m} J(A^*_1 \cup A^*_2 \cup \cdots A^*_i).
\]
\end{lemma}

\begin{proof}
From Lemma \ref{lem:submodular_chain}, we have 
\begin{align*}
(m+1)J(V) - \sum_{i=1}^{m} J(A_1 \cup A_2 \cup \cdots \cup A_i)  &\leq mJ(V) - 2(1 - e^{-\tau})\sum_{i=1}^{m/2-1}J(A^*_1 \cup A^*_2 \cup \cdots A^*_i) \\
&\leq e^{-\tau}mJ(V) + mJ(V) - 2\sum_{i=1}^{m/2 -1} J(A^*_1 \cup A^*_2 \cup \cdots A^*_i) \\
&\leq e^{-\tau}mJ(V) +  2\sum_{i=1}^{m/2 -1}J(V) - J(A^*_1 \cup A^*_2 \cup \cdots A^*_i).
\end{align*}
Now, we use the monotonicity property of $J$ to get
\begin{align*}
e^{-\tau}mJ(V) + 2\sum_{i=1}^{m/2 -1}J(V) - J(A^*_1 \cup A^*_2 \cup \cdots A^*_i) &\leq e^{-\tau}mJ(V) + 2\sum_{i=1}^{m}J(V) - J(A^*_1 \cup A^*_2 \cup \cdots A^*_i).
\end{align*}
Finally, we have
\begin{align*}
\sum_{i=1}^m J(V) - J(A_1 \cup A_2 \cup \cdots \cup A_i)  &\leq (m+1)J(V) - \sum_{i=1}^m J(A_1 \cup A_2 \cup \cdots \cup A_i) \\
&\leq e^{-\tau} m \cdot J(V) + 2\sum_{i=1}^{m}J(V) - J(A^*_1 \cup A^*_2 \cup \cdots A^*_i).
\end{align*}
\end{proof}

\begin{lemma}\label{lem:zero_wt_opt}
Suppose $S^+$ denote optimal separating set system that uses an additional color of weight $1$ and uses weight $0$ vector to color $A_0$. Let $S^*$ denote optimal separating system. Then, we have $C(S^+) - C(S^*) \leq 0$.
\end{lemma}
\begin{proof}
We give a proof similar to Lemma $22$ in~\cite{neurips18}. Let the set of nodes $A_0$ selected in the first iteration of greedy Algorithm~\ref{alg:ss_matrix} and assigned weight $0$ vector be denoted by $S^+_0$. Similarly, the set of nodes that are colored with weight $0$ vector in $S^*$ be denoted by $S_0^*$. As $S^+$ denotes optimal solution on $V \setminus S^+_0$, we can assume that a solution that uses a weight $1$ color for nodes in $S_0^* \setminus S_0^+$ is only going to be worse. Therefore, we have : 
\begin{align*}
    C(S^+) - C(S^*) &\leq \sum_{v \in S_0^* \setminus S_0^+} C(v) - \sum_{v \in S_0^+ \setminus S_0^*} C(v) \\
                    &\leq \sum_{v \in S_0^* \setminus S_0^+} C(v) + \sum_{v \in S_0^* \cap S_0^+} C(v) - \sum_{v \in S_0^+ \setminus S^*_0} C(v) - \sum_{v \in S_0^* \cap S_0^+} C(v)\\
                    &\leq \sum_{v \in S_0^*} C(v) - \sum_{v \in S_0^+} C(v) \leq 0.
\end{align*}
\end{proof}

Let $L_{\eps}$ denote the $\eps$-separating matrix returned by Algorithm~\ref{alg:ss_matrix}, and let $L_{\eps} = \{A_0, A_1, A_2, \cdots A_m \}$ where we abuse the previous notation and denote $A_i$ to represent the set of all nodes (instead of a collection of subsets of $V$) that have weight $i$ assigned by $L_{\eps}$. 
\[
C(L_{\eps}) = \sum_{i=1}^m J(V) - J(A_1 \cup A_2 \cup \cdots A_i),
\]
where $ |L_i|  \leq {m \choose i}$. We observe that this cost representation corresponds to the supermodular chain problem discussed above.

Assuming $m \geq \eta \log 1/\eps$ for $\eta > 2$, we have that, the greedy Algorithm~\ref{alg:ss_matrix} only uses vectors of weight at most $\log 1/\eps$ i.e., $m/2$. In Lemma~\ref{lem:submodular_chain}, each of the values $k_1, k_2 \cdots k_i \cdots k_p$ represent number of weight $i$ vectors available and from Lemma 21 in~\citet{neurips18}, we have that $\tau = \Omega(m)$, for every $p$ in the range. 

\begin{lemma}(Lemma~\ref{lem:2approx} restated)
For any $m \geq \eta \log 1/\epsilon$ for some constant $\eta > 2$, with probability $\ge 1-\delta$, Algorithm~\ref{alg:ss_matrix} returns $L_\epsilon$ with $C(L_\epsilon) \leq (2 + \exp{(-\Omega(m))})  \cdot C(L^*)$, where $L^*$ is the min-cost separating matrix for $G$.
\end{lemma}
\begin{proof} Using the definition of $S^+$ from Lemma~\ref{lem:zero_wt_opt}, we argue that $C(S^+) \geq J(V)$  as every node in $V$ is assigned a weight $1$ vector. From Lemma \ref{lem:supermodular-chain}, we have
\begin{align*}
\mathrm{C}(L_{\eps}) &= \sum_{i=1}^{m/2} J(V) - J(A_1 \cup A_2 \cup \cdots A_i) \\
&\leq e^{-\tau} m \cdot J(V) + 2\sum_{i=1}^{m}J(V) - J(A^*_1 \cup A^*_2 \cup \cdots A^*_i) \\
&\leq e^{-\tau} m \cdot C(S^+) + 2\sum_{i=1}^{m}J(V) - J(A^*_1 \cup A^*_2 \cup \cdots A^*_i) \\
&\leq e^{-\tau} m \cdot C(S^+) + 2 C(S^*).\\
\text{From Lemma~\ref{lem:zero_wt_opt}, we have } \\
&\leq (2 + \exp{(-\Omega(m))}) \cdot C(S^*) = (2 + \exp{(-\Omega(m))}) \cdot C(L^*).
\end{align*}
\end{proof}

\end{document}